\documentclass[12pt]{amsart}
\usepackage[a4paper,top=2.5cm,bottom=2cm,left=1.5cm,right=1.5cm,marginparwidth=1.5cm]{geometry}
\usepackage{amsmath,amssymb,mathtools}
\usepackage{enumitem}
\usepackage{graphicx}
\usepackage{subcaption}
\usepackage{xcolor}
\usepackage{array}
\usepackage{booktabs}
\usepackage{verbatim}
\usepackage{float}
\usepackage{bbm}
\usepackage[format=plain,font=it]{caption}
\usepackage[numbers,sort&compress]{natbib}
\usepackage{hyperref}
\usepackage[nameinlink,capitalize]{cleveref}

\hypersetup{
  colorlinks=true,
  linkcolor=blue!55!black,
  citecolor=red,
  urlcolor=blue!55!black
}

\theoremstyle{plain}
\newtheorem{theorem}{Theorem}

\newtheorem{lemma}{Lemma}
\newtheorem{corollary}{Corollary}

\theoremstyle{definition}
\newtheorem{assumption}{Assumption}
\newtheorem{definition}{Definition}

\theoremstyle{remark}
\newtheorem{remark}{Remark}

\DeclareMathOperator{\alg}{alg}
\DeclareMathOperator{\geo}{geo}
\DeclareMathOperator{\rank}{rank}
\DeclareMathOperator{\adj}{adj}
\DeclareMathOperator{\tr}{tr}

\title[Subwavelength Exceptional Points in Dispersive Resonator Arrays]
{Subwavelength Exceptional Points in Dispersive Resonator Arrays}

\author{Konstantinos Alexopoulos}
\address{MAP5, Université Paris Cité, CNRS, Paris, France}
\email{kalexopoulos@math.cnrs.fr}
\date{}

\begin{document}

\begin{abstract}
Exceptional points provide a mechanism for producing strong spectral sensitivity in non-Hermitian wave systems. For subwavelength resonators with highly dispersive material parameters, however, the resonance problem is not a linear eigenvalue problem in the frequency. In this paper, we use the Lippmann--Schwinger formulation and the associated volume-potential reduction to study exceptional points in dispersive resonator arrays. The reduction leads to a finite-dimensional nonlinear matrix pencil whose entries encode both the local material response and the volume-mediated interaction between the resonators. We characterize exceptional points as defective characteristic values of this pencil and derive explicit conditions for dimers and higher-order systems. We show that passive damping and detuning can generate exceptional points without imposed gain, and we formulate an inverse-design condition for placing such degeneracies at prescribed frequencies. We also prove the corresponding Puiseux splitting laws and identify how material dispersion enters the leading sensitivity prefactor. Numerical experiments, including halide-perovskite-inspired dimers and trimers, illustrate the theory and show that the exceptional-point mechanism persists for explicit highly dispersive passive material laws.
\end{abstract}

\maketitle

\section{Introduction}
\label{sec:introduction}

Subwavelength resonator arrays provide a powerful mechanism for manipulating waves at scales much smaller than the operating wavelength. Their behaviour is governed by the interaction between many small components, and this interaction can often be described by a finite-dimensional matrix obtained from the underlying scattering problem. This point of view has been developed through layer-potential and volume-potential methods for acoustic and electromagnetic resonators \cite{ammari2009layer,ammari2018mathematical,ammari2022wave,ammari2024functional}. It has been used to describe hybridized resonances, closely spaced resonator effects, frequency separation and graded-array behaviour \cite{ammari2020close,ammari2019fully}, as well as PT-symmetric subwavelength structures and exceptional-point sensing \cite{ammari2021high,ammari2022exceptional}.

A central objective in resonator-array design is to create spectral responses which are strong, tunable and sensitive to small changes of the system. Exceptional points offer one way of producing such behaviour. At an exceptional point, resonant frequencies and resonant modes coalesce, and the usual linear perturbation picture breaks down. Near an exceptional point of order $m$, the resonance branches typically split according to an $m$th-root law. This singular response has made exceptional points central objects in non-Hermitian wave physics and photonics \cite{heiss2012physics,miri2019exceptional,el2018non,li2023exceptional}, with important connections to PT-symmetric models \cite{bender1998real}, non-Hermitian spectral theory \cite{moiseyev2011non}, and analytic perturbation theory \cite{kato1966perturbation}.

This mechanism is particularly attractive for sensing. A square-root splitting, for instance, can produce a larger frequency separation than the one obtained near a simple resonance. This idea has been explored in optical microcavities and higher-order photonic systems \cite{wiersig2014enhancing,hodaei2017enhanced,chen2017exceptional}. At the same time, the interpretation of exceptional-point-enhanced sensing is delicate: linewidth, noise, modal non-orthogonality and the measurement procedure all affect the observable signal \cite{lau2018fundamental,wiersig2020review,wiersig2020prospects}. It is therefore useful to have reduced models in which the degeneracy conditions and the splitting constants can be computed directly from the physical parameters. Nanoscale non-Hermitian systems, including plasmonic realizations of exceptional points, give one concrete setting where this question is important \cite{park2020symmetry,li2023exceptional}.

The difficulty addressed in the present paper is that, for highly dispersive resonators, this reduced matrix is not a standard eigenvalue matrix. The material parameters depend on the frequency, sometimes through singular or strongly varying laws, and the resonance equation becomes a nonlinear characteristic-value problem. Halide perovskites are a motivating class of such materials, since they have tunable optical properties and are widely studied in photonic and optoelectronic applications \cite{sutherland2016perovskite,chouhan2020synthesis}. In previous work, the Lippmann--Schwinger formulation was used to derive asymptotic resonance equations for highly dispersive halide-perovskite resonators and for coupled highly dispersive systems \cite{alexopoulos2022asymptotic,alexopoulos2023mathematical}. The graded-array reduction separates the diagonal dispersive response from the off-diagonal volume-mediated interaction and turns this matrix into a design object for frequency-position control \cite{alexopoulos2026graded}. Further work on high-index nanoparticles shows how closely spaced electromagnetic resonators can generate strong near-field enhancement without a singularity \cite{alexopoulos2026non}.

The purpose of this paper is to understand what becomes of exceptional points in this nonlinear dispersive setting. The main issue is not only to solve for coalescing resonances, but to identify where the material dispersion enters the exceptional-point equations and the associated sensitivity law. We define exceptional points as defective characteristic values of the reduced nonlinear matrix pencil. In this formulation, material dispersion affects the frequency derivatives of the characteristic determinant, and hence changes the conditions which impose multiple characteristic values. This gives a way to distinguish the effect of the material values from the effect of the material derivatives, to compute the Puiseux prefactor in terms of the reduced pencil, and to formulate inverse design as a transversality problem. The questions we address are the following. How can one detect exceptional points in a dispersive reduced matrix? How do the nearby resonances split under perturbations? Can passive damping and detuning generate exceptional points without imposed gain? Finally, can the same determinant-based method be used to construct higher-order exceptional points?

In this work, we start in Section~\ref{sec:reduced-problem} by recalling the Helmholtz problem, the Lippmann--Schwinger representation and the finite-dimensional reduction which gives the nonlinear matrix pencil $\mathcal L(\omega,p)$. In Section~\ref{sec:nonlinear-pencils}, we define characteristic values, algebraic and geometric multiplicities, and exceptional points for nonlinear matrix pencils, and we isolate the role of material dispersion through a frozen-material comparison and a local inverse-design theorem. In Section~\ref{sec:dimers}, we apply these definitions to dimers and derive explicit second-order exceptional-point conditions. In Section~\ref{sec:perturbation}, we prove the Puiseux splitting law near an exceptional point and identify the sensitivity prefactor in terms of derivatives of the reduced pencil. In Section~\ref{sec:higher-order}, we extend the reduced theory to higher-order systems and prove a third-order trimer construction. Finally, in Section~\ref{sec:numerics}, we test the reduced theory numerically for dimers and trimers, including passive exceptional points and inverse design. The numerical section ends with halide-perovskite-inspired dimer and full-interaction trimer computations, which show that the same exceptional-point signatures persist for an explicit highly dispersive material law.

\section{The Reduced Dispersive Resonator Problem}
\label{sec:reduced-problem}

Let us begin by introducing the finite-dimensional setting on which the exceptional-point analysis will be carried out. We recall the scattering problem and show how the dispersive material law and the volume-mediated interaction enter the same reduced nonlinear pencil.

\subsection{Helmholtz formulation}

Let us first describe the physical system which will be studied throughout the paper. Let $d\in\{2,3\}$ and consider an array of $N$ small resonators
\begin{align*}
  \Omega_\delta
  =
  \bigcup_{i=1}^N \Omega_i^\delta,
  \qquad
  \Omega_i^\delta = z_i+\delta D_i,
  \qquad 0<\delta\ll 1,
\end{align*}
where the reference domains $D_i\subset\mathbb R^d$ are bounded, connected and have Lipschitz boundaries. The centres $z_i$ and the domains $D_i$ may depend on the index $i$, allowing for geometric grading. We assume throughout that the components are disjoint and remain separated in the rescaled variables:
\begin{align*}
  \operatorname{dist}(D_i,D_j)\geq d_* >0,
  \qquad i\neq j,
\end{align*}
with $d_*$ independent of $\delta$.

The background medium has permittivity $\varepsilon_0$ and permeability $\mu_0$. We denote by $k_0$ the corresponding background wavenumber, namely
\begin{align*}
  k_0 = \omega\sqrt{\varepsilon_0\mu_0}.
\end{align*}
Inside the $i$th resonator, the permittivity is denoted by $\varepsilon_i(\omega)$. We allow $\varepsilon_i$ to be complex-valued and frequency-dependent. Its imaginary part models damping or gain, while its frequency dependence models material dispersion. We write the corresponding contrast factor as
\begin{align*}
  \xi_i(\omega)
  :=
  \mu_0\bigl(\varepsilon_i(\omega)-\varepsilon_0\bigr).
\end{align*}

\begin{remark}
The important point for what follows is that $\xi_i$ is not a constant contrast. It is a frequency-dependent, and in general complex-valued, function. For highly dispersive materials it may also be meromorphic in the frequency. This is the source of the nonlinear characteristic-value problem which appears after the subwavelength reduction.
\end{remark}

For an incident field $u^{\mathrm{in}}$ satisfying
\begin{align*}
  (\Delta+k_0^2)u^{\mathrm{in}}=0
  \qquad\text{in }\mathbb R^d,
\end{align*}
the total field $u$ solves
\begin{align*}
  \begin{cases}
    \Delta u+\omega^2\varepsilon_i(\omega)\mu_0 u=0,
      & x\in \Omega_i^\delta,\quad i=1,\ldots,N,\\
    \Delta u+k_0^2u=0,
      & x\in \mathbb R^d\setminus\overline{\Omega_\delta},\\
    u|_+=u|_-,
      & x\in\partial\Omega_\delta,\\
    \partial_\nu u|_+=\partial_\nu u|_-,
      & x\in\partial\Omega_\delta,\\
    u-u^{\mathrm{in}}\text{ satisfies the outgoing radiation condition.}
  \end{cases}
\end{align*}
The outgoing radiation condition is the Sommerfeld condition
\begin{align*}
  \lim_{r\to\infty}
  r^{(d-1)/2}
  \left(
    \frac{\partial}{\partial r}-ik_0
  \right)
  \bigl(u-u^{\mathrm{in}}\bigr)
  =
  0,
  \qquad r=|x|,
\end{align*}
uniformly in the angular variable. After the change of variables $x=z_i+\delta y$ on each component, the subwavelength regime corresponds to
\begin{align*}
  \delta k_0 \ll 1.
\end{align*}
\subsection{Lippmann--Schwinger formulation}

Let $G(\cdot,k)$ be the outgoing Helmholtz Green function, defined by
\begin{align*}
  (\Delta+k^2)G(\cdot,k)=\delta_0
  \qquad\text{in }\mathbb R^d,
\end{align*}
together with the outgoing radiation condition. With the sign convention used here,
\begin{align*}
  G(x,k)
  =
  \begin{cases}
    -\dfrac{i}{4}H_0^{(1)}(k|x|),
      & d=2,\\[1em]
    -\dfrac{e^{ik|x|}}{4\pi |x|},
      & d=3.
  \end{cases}
\end{align*}
\begin{theorem}[Lippmann--Schwinger representation]
\label{thm:lippmann-schwinger}
The field $u$ solves the rescaled Helmholtz transmission problem if and only if
it satisfies
\begin{align*}
  u(x)-u^{\mathrm{in}}(x)
  =
  -\delta^2\omega^2
  \sum_{i=1}^N
  \xi_i(\omega)
  \int_{D_i}
    G(x-y,\delta k_0)u(y)\,dy,
  \qquad x\in\mathbb R^d.
\end{align*}
In particular, resonant modes are non-trivial solutions of the same equation with $u^{\mathrm{in}}=0$.
\end{theorem}

\begin{proof}
The detailed proof of the statement can be found in \cite{alexopoulos2022asymptotic, alexopoulos2023mathematical}. A sketch of the proof is provided in Appendix~\ref{app:proof-ls}.
\end{proof}

Let us also introduce the componentwise volume-potential operators
\begin{align*}
  \mathcal K_i^{\delta k_0}[\varphi](x)
  :=
  -\int_{D_i}G(x-y,\delta k_0)\varphi(y)\,dy,
  \qquad x\in D_i,
\end{align*}
and, for $i\neq j$,
\begin{align*}
  \mathcal R_{ji}^{\delta k_0}[\varphi](x)
  :=
  -\int_{D_j}G(x-y,\delta k_0)\varphi(y)\,dy,
  \qquad x\in D_i.
\end{align*}
The first operator describes the self-interaction of $D_i$, while the second describes the field induced on $D_i$ by a density supported on $D_j$.

\subsection{Finite-dimensional resonance equation}

We now pass from the integral equation to the reduced matrix problem. The finite-dimensional model is obtained by projecting the volume-potential equation onto the dominant resonant direction of each component. The following statement records the one-mode approximation which is used in the reduction.

\begin{lemma}[Dominant resonant profile]
\label{lem:dominant-profile}
Assume that, for each $i$, the self-interaction operator
$\mathcal K_i^{\delta k_0}$ has a simple isolated eigenpair
\begin{align*}
  \mathcal K_i^{\delta k_0}\phi_i
  =
  \nu_i(\delta,\omega)\phi_i,
  \qquad
  \|\phi_i\|_{L^2(D_i)}=1,
\end{align*}
in the frequency window of interest. In the constant-mode regime used in the graded-array reduction, $\phi_i$ is approximated by
\begin{align*}
  \widehat 1_{D_i}
  :=
  |D_i|^{-1/2}1_{D_i}.
\end{align*}
More precisely, in the two-dimensional asymptotic regime,
\begin{align*}
  \phi_i
  =
  \widehat 1_{D_i}
  +
  O\bigl(|\log\delta|^{-1}\bigr),
  \qquad i=1,\ldots,N.
\end{align*}
\end{lemma}

\begin{proof}
The detailed proof of the statement can be found in \cite{alexopoulos2022asymptotic, alexopoulos2023mathematical}. A sketch of the proof is provided in Appendix~\ref{app:proof-dominant-profile}.
\end{proof}

The resonant mode can therefore be decomposed, at leading order, as
\begin{align*}
  u|_{D_i}\approx q_i\phi_i,
  \qquad i=1,\ldots,N,
\end{align*}
where the coefficients $q_i$ describe the contribution of the $i$th resonator to the collective mode.

\begin{lemma}[Projected resonance system]
\label{lem:projected-system}
Projecting the homogeneous Lippmann--Schwinger equation onto the dominant
profiles $\phi_1,\ldots,\phi_N$ gives, at leading order, the nonlinear
finite-dimensional system
\begin{align*}
  \mathcal L(\omega,p) q = 0,
\end{align*}
where $q=(q_1,\ldots,q_N)^\top\in\mathbb C^N$ and $p$ denotes the tunable parameters of the array. In the constant-mode approximation used for graded arrays, the reduced matrix has the schematic form
\begin{align*}
  \mathcal L(\omega,p)
  =
  I-\operatorname{diag}\bigl(B_1(\omega,p),\ldots,B_N(\omega,p)\bigr)C(\omega,p),
\end{align*}
where the scalar functions $B_i$ encode the local dispersive response and the matrix $C$ contains the projected volume-mediated coupling.
\end{lemma}

\begin{proof}
The detailed proof of the statement can be found in \cite{alexopoulos2026graded, alexopoulos2022asymptotic, alexopoulos2023mathematical}. A sketch of the proof is provided in Appendix~\ref{app:proof-projected-system}.
\end{proof}

We now collect the hypotheses under which the reduced matrix problem is used as the leading-order description of the original scattering problem. The assumption below isolates the one-mode regime and the uniform spectral separation needed for the projection argument to remain stable throughout the frequency window.

\begin{assumption}[Subwavelength one-mode reduction]
\label{ass:one-mode-reduction}
In the frequency window $W\subset\mathbb C$ under consideration, we assume the
following.
\begin{enumerate}[label=(\roman*)]
  \item The number of resonators $N$ is fixed. There exist constants
  $d_*>0$ and $\delta_*>0$ such that
  \begin{align*}
    \operatorname{dist}(D_i,D_j)\geq d_*,
    \qquad i\neq j,
    \qquad 0<\delta<\delta_*.
  \end{align*}
  Moreover, if $\mathcal P_i(p)$ denotes the set of poles of
  $\xi_i(\cdot,p)$ which are not part of the resonant denominators under
  study, then, there exists $\rho_*>0$ such that
  \begin{align*}
    \operatorname{dist}\bigl(W,\mathcal P_i(p)\bigr)\geq \rho_*,
    \qquad i=1,\ldots,N,
  \end{align*}
  uniformly for the admissible parameter values $p$.
  \item For each component $D_i$, the self-interaction operator
  $\mathcal K_i^{\delta k_0}$ has a simple isolated eigenvalue
  $\nu_i(\delta,\omega)$ with normalized eigenfunction $\phi_i$. If
  $\sigma_i^{\mathrm{rem}}(\delta,\omega)$ denotes the rest of the spectrum
  of $\mathcal K_i^{\delta k_0}$, then, there exist constants $c_*>0$ and
  $\delta_*>0$ such that
  \begin{align*}
    \operatorname{dist}
    \bigl(
      \nu_i(\delta,\omega),
      \sigma_i^{\mathrm{rem}}(\delta,\omega)
    \bigr)
    \geq c_*,
  \end{align*}
  for all $i=1,\ldots,N$, all $\omega\in W$, and all
  $0<\delta<\delta_*$.
  \item The resonant field admits a one-mode decomposition
  \begin{align*}
    u|_{D_i}
    =
    q_i\phi_i+r_i,
    \qquad
    \|r_i\|_{L^2(D_i)}
    \leq
    e_\delta \|q\|,
  \end{align*}
  where $e_\delta\to0$ as $\delta\to0$, uniformly for $\omega\in W$. In the
  two-dimensional constant-profile regime used in the graded-array reduction,
  one may take the leading profile to be
  $\widehat 1_{D_i}=|D_i|^{-1/2}1_{D_i}$, with profile error
  $O(|\log\delta|^{-1})$, together with the many-particle reduction error.
\end{enumerate}
\end{assumption}

For fixed parameters $p$, we denote the reduced characteristic determinant by
\begin{align*}
  F(\omega,p):=\det\mathcal L(\omega,p).
\end{align*}

\begin{theorem}[Reduced resonance equation]
\label{thm:reduced-resonance-equation}
Let us assume that Assumption~\ref{ass:one-mode-reduction} holds. Then, the
full resonant equation in the window $W$ can be written, after projection onto
the dominant component modes, as
\begin{align*}
  \mathcal L(\omega,p)q+\mathcal E_\delta(\omega,p)q=0,
  \qquad
  \|\mathcal E_\delta(\omega,p)\|\leq C e_\delta,
\end{align*}
uniformly for $\omega\in W$ and for admissible parameter values $p$. Hence the leading-order resonances of the full system are given by the characteristic values of the nonlinear matrix problem
\begin{align*}
  \mathcal L(\omega,p)q=0.
\end{align*}
Moreover, each simple zero of $F(\cdot,p)$ persists as a unique nearby resonance of the full projected problem, counted with multiplicity.
\end{theorem}

\begin{proof}
The complete proof of this statement can be found in \cite{alexopoulos2022asymptotic, alexopoulos2023mathematical, alexopoulos2026graded}. A sketch of the proof is provided in Appendix~\ref{app:proof-reduced}.
\end{proof}

\subsection{Nonlinear character of the reduced pencil}

The key point for the present paper is that $\mathcal L$ is not a matrix of the form $A(p)-\omega I$. Its entries depend on $\omega$ through the dispersive material factors $\xi_i(\omega)$ and through the Helmholtz volume-potential interactions. Here, non-affine dependence means dependence which cannot be written in the form $M_0(p)+\omega M_1(p)$ on the frequency window under consideration. Thus, the reduced problem is a nonlinear characteristic-value problem. 

\begin{lemma}[Material dispersion induces a nonlinear pencil]
\label{lem:dispersion-nonlinear-pencil}
Assume that, in the frequency window $W$, each material contrast
$\xi_i(\cdot,p)$ is holomorphic and that the projected self-interaction and
cross-interaction terms are holomorphic in $\omega$. Then, the reduced matrix
$\mathcal L(\cdot,p)$ is a holomorphic matrix-valued function on $W$. If at
least one entry of $\mathcal L$ depends non-affinely on $\omega$ through a
material contrast or through a frequency-dependent volume interaction, then,
the reduced resonance equation
\begin{align*}
  \mathcal L(\omega,p)q=0
\end{align*}
is a nonlinear characteristic-value problem.
\end{lemma}

\begin{proof}
The proof of this statement is provided in Appendix~\ref{app:proof-dispersion-pencil}.
\end{proof}

\begin{remark}
    The resonance condition in Lemma \ref{lem:dispersion-nonlinear-pencil} is not equivalent, in general, to a linear eigenvalue problem $A(p)q=\omega B(p)q$. If the entries were equivalent to a linear pencil $A(p)-\omega B(p)$, then, each entry would be affine in $\omega$. A non-affine dependence inherited from a material contrast or from a frequency-dependent interaction contradicts this structure. Hence the reduced resonance equation is, in general, a nonlinear characteristic-value problem.
\end{remark}

\section{Exceptional Points of Nonlinear Matrix Pencils}
\label{sec:nonlinear-pencils}

Let us now formulate exceptional points in the nonlinear setting introduced in the previous section.


\subsection{Characteristic values}

Throughout this section the parameter $p$ is fixed, and $\mathcal L(\cdot,p)$ is regarded as a matrix-valued function of the complex frequency. The word pencil refers to this frequency-dependent matrix-valued function, not necessarily to a linear pencil of the form $A-\omega B$. We work away from the poles of the dispersive material law, so that $\mathcal L(\cdot,p)$ is holomorphic in the frequency window under consideration. We also assume that the pencil is regular in this window, which means that $\det\mathcal L(\cdot,p)$ is not identically zero.

\begin{definition}[Characteristic value]
\label{def:characteristic-value}
A frequency $\omega_*$ is a characteristic value of $\mathcal L(\cdot,p)$ if
\begin{align*}
  \ker \mathcal L(\omega_*,p)\neq \{0\}.
\end{align*}
Equivalently, if
\begin{align*}
  F(\omega,p):=\det \mathcal L(\omega,p),
\end{align*}
then, $\omega_*$ is a characteristic value if and only if
\begin{align*}
  F(\omega_*,p)=0.
\end{align*}
\end{definition}

The \emph{geometric multiplicity} of a characteristic value is defined by
\begin{align*}
  \geo_{\omega_*}\mathcal L(\cdot,p)
  :=
  \dim\ker \mathcal L(\omega_*,p).
\end{align*}
This is the number of linearly independent reduced resonant modes associated with the frequency $\omega_*$.

\subsection{Algebraic multiplicity}

Since the reduced problem is finite-dimensional, the algebraic multiplicity can be read from the determinant. We say that $\omega_*$ has \emph{algebraic multiplicity} $m$ if $F(\cdot,p)$ has a zero of order $m$ at $\omega_*$. Thus,
\begin{align*}
  F(\omega,p)
  =
  c(\omega-\omega_*)^m
  +
  O\bigl((\omega-\omega_*)^{m+1}\bigr),
  \qquad c\neq 0.
\end{align*}
Equivalently,
\begin{align*}
  F(\omega_*,p)
  =
  \partial_\omega F(\omega_*,p)
  =
  \cdots
  =
  \partial_\omega^{m-1}F(\omega_*,p)
  =
  0,
\end{align*}
while
\begin{align*}
  \partial_\omega^mF(\omega_*,p)\neq 0.
\end{align*}
We denote this integer by
\begin{align*}
  \alg_{\omega_*}\mathcal L(\cdot,p).
\end{align*}

\begin{remark}
For a holomorphic matrix-valued function, this determinant-based definition is the finite-dimensional version of the algebraic multiplicity of a characteristic value in the sense of Gohberg--Sigal theory \cite{gohberg1971operator}. This is the notion used in subwavelength resonance theory when resonances are formulated as characteristic values of operator-valued functions.
\end{remark}

\subsection{Exceptional points}

\begin{definition}[Exceptional point of the reduced system]
\label{def:reduced-exceptional-point}
Let $\omega_*$ be a characteristic value of $\mathcal L(\cdot,p_*)$. We say
that $(\omega_*,p_*)$ is an exceptional point of order $m$ if
\begin{align*}
  \alg_{\omega_*}\mathcal L(\cdot,p_*) = m,
  \qquad
  \geo_{\omega_*}\mathcal L(\cdot,p_*) = 1.
\end{align*}
\end{definition}

In terms of the determinant, $(\omega_*,p_*)$ is therefore an exceptional point of order $m$ if
\begin{align}
  F(\omega_*,p_*)
  =
  \partial_\omega F(\omega_*,p_*)
  =
  \cdots
  =
  \partial_\omega^{m-1}F(\omega_*,p_*)
  =
  0,
  \label{eq:ep-derivative-conditions}
\end{align}
\begin{align}
  \partial_\omega^mF(\omega_*,p_*)\neq 0,
  \qquad
  \dim\ker \mathcal L(\omega_*,p_*)=1.
  \label{eq:ep-nondegeneracy-conditions}
\end{align}

\begin{lemma}[Determinant criterion for reduced exceptional points]
\label{lem:determinant-criterion}
Assume that $\mathcal L(\cdot,p_*)$ is holomorphic near $\omega_*$ and that
$F(\omega,p_*)=\det\mathcal L(\omega,p_*)$ has a zero of order $m$ at
$\omega_*$. If
\begin{align*}
  \dim\ker \mathcal L(\omega_*,p_*)=1,
\end{align*}
then, $(\omega_*,p_*)$ is an exceptional point of order $m$ of the reduced dispersive resonator system.
\end{lemma}

\begin{proof}
The proof of this statement is provided in Appendix~\ref{app:proof-det-criterion}.
\end{proof}

\begin{remark}
If $\mathcal L(\omega,p)=A(p)-\omega I$, then, the definition above reduces to the usual notion of an exceptional point for a finite-dimensional non-Hermitian matrix. In the present setting the same defect occurs, but for a nonlinear characteristic equation. The derivatives of $F$ therefore contain contributions from both the dispersive material factors and the frequency-dependent volume interactions.
\end{remark}

\subsection{Material dispersion in the exceptional-point conditions}
\label{subsec:dispersion-ep-conditions}

The determinant conditions \eqref{eq:ep-derivative-conditions} make the role of dispersion explicit. Indeed, by Jacobi's formula,
\begin{align*}
  \partial_\omega F(\omega,p)
  =
  \tr\bigl(
    \adj\mathcal L(\omega,p)
    \partial_\omega\mathcal L(\omega,p)
  \bigr).
\end{align*}
Since $\partial_\omega\mathcal L$ contains the material derivatives $\partial_\omega\xi_i$ as well as the frequency derivatives of the interaction terms, the equations defining a multiple characteristic value are changed by the dispersive material law itself, and not only by the values of the material parameters at the exceptional frequency.

To isolate this effect, fix a candidate exceptional point $(\omega_*,p_*)$ and define the frozen-material pencil $\mathcal L_{\mathrm{fr}}$ by replacing each material contrast by its value at $\omega_*$,
\begin{align*}
  \xi_i(\omega,p)
  \quad\longmapsto\quad
  \xi_i(\omega_*,p),
\end{align*}
while keeping the same geometric and interaction model. Thus, $\mathcal L_{\mathrm{fr}}$ has the same material values at $\omega_*$, but it does not contain the material derivatives $\partial_\omega\xi_i(\omega_*,p)$.

\begin{lemma}[Effect of freezing the material response]
\label{lem:frozen-material}
Assume that $\mathcal L$ and $\mathcal L_{\mathrm{fr}}$ are holomorphic near
$(\omega_*,p_*)$ and that
\begin{align*}
  \mathcal L_{\mathrm{fr}}(\omega_*,p_*)
  =
  \mathcal L(\omega_*,p_*).
\end{align*}
Then,
\begin{align*}
  \partial_\omega F(\omega_*,p_*)
  -
  \partial_\omega F_{\mathrm{fr}}(\omega_*,p_*)
  =
  \tr\Bigl(
    \adj\mathcal L(\omega_*,p_*)
    \bigl(
      \partial_\omega\mathcal L
      -
      \partial_\omega\mathcal L_{\mathrm{fr}}
    \bigr)(\omega_*,p_*)
  \Bigr),
\end{align*}
where $F=\det\mathcal L$ and $F_{\mathrm{fr}}=\det\mathcal L_{\mathrm{fr}}$. In particular, the double-root equation for the dispersive pencil differs from the frozen-material equation by terms containing the material derivatives $\partial_\omega\xi_i(\omega_*,p_*)$.
\end{lemma}

\begin{proof}
The proof of this statement is provided in Appendix~\ref{app:proof-frozen-material}.
\end{proof}

\begin{remark}
The frozen-material pencil is not introduced as a competing physical model. It is a diagnostic comparison. If the material law is frozen, the values of the contrast at $\omega_*$ are retained, but the material contribution to $\partial_\omega\mathcal L$ is removed. Thus, the comparison separates the effect of the material values from the effect of material dispersion.
\end{remark}

\subsection{Reduced-to-full interpretation}

\begin{corollary}[Persistence of an exceptional resonance cluster]
\label{cor:exceptional-cluster}
Let us assume that Assumption~\ref{ass:one-mode-reduction} holds. Let
$(\omega_*,p_*)$ be an order-$m$ exceptional point of the reduced matrix
$\mathcal L$, and assume that $\omega_*$ is the only zero of $F(\cdot,p_*)$ in
$\overline{B(\omega_*,r)}\subset W$. Then, for all sufficiently small
$\delta$, the full projected determinant
\begin{align*}
  F_\delta(\omega,p_*)
  :=
  \det\bigl(\mathcal L(\omega,p_*)+\mathcal E_\delta(\omega,p_*)\bigr)
\end{align*}
has exactly $m$ zeros in $B(\omega_*,r)$, counted with multiplicity. Moreover, these zeros converge to $\omega_*$ as $\delta\to0$.
\end{corollary}

\begin{proof}
The proof of this statement is provided in Appendix~\ref{app:proof-cluster}.
\end{proof}

\begin{remark}
Corollary~\ref{cor:exceptional-cluster} is the reduced-to-full statement used in this paper. The determinant equations are solved for the leading matrix $\mathcal L$. The corresponding full projected problem does not necessarily have an exact multiple zero at the same value of $\omega$ for finite $\delta$, but it has a cluster of $m$ nearby resonances with the same total multiplicity. Thus, the reduced exceptional point describes the leading position of the coalescing resonance cluster of the original Lippmann--Schwinger problem.
\end{remark}

\subsection{Local inverse design of second-order exceptional points}
\label{subsec:inverse-design-theorem}

For second-order exceptional points, the conditions \eqref{eq:ep-derivative-conditions}--\eqref{eq:ep-nondegeneracy-conditions} reduce to the two complex equations
\begin{align*}
  F(\omega,p)=0,
  \qquad
  \partial_\omega F(\omega,p)=0,
\end{align*}
together with the nondegeneracy conditions $\partial_\omega^2F(\omega,p)\neq0$ and $\dim\ker\mathcal L(\omega,p)=1$. This gives a natural local inverse-design problem: tune two real active parameters in order to impose the real and imaginary parts of these two equations.

\begin{theorem}[Transversality and local inverse design]
\label{thm:transversality-inverse-design}
Let $p=(a,b)\in\mathbb R^2\times\mathbb R^s$, where $a$ denotes two active
design parameters and $b$ denotes the remaining parameters. Suppose that
$\mathcal L(\omega,a,b)$ is holomorphic in $\omega$ and continuously
differentiable in $(a,b)$ near a second-order exceptional point
$(\omega_*,a_*,b_*)$. Define
\begin{align*}
  \mathcal T(\omega,a,b)
  :=
  \left(
    \operatorname{Re}F,
    \operatorname{Im}F,
    \operatorname{Re}\partial_\omega F,
    \operatorname{Im}\partial_\omega F
  \right)(\omega,a,b).
\end{align*}
If the real Jacobian
\begin{align*}
  D_{(\operatorname{Re}\omega,\operatorname{Im}\omega,a)}
  \mathcal T(\omega_*,a_*,b_*)
\end{align*}
is invertible, then, for $b$ sufficiently close to $b_*$, there exist unique functions $\omega(b)$ and $a(b)$ near $\omega_*$ and $a_*$ such that
\begin{align*}
  F(\omega(b),a(b),b)=0,
  \qquad
  \partial_\omega F(\omega(b),a(b),b)=0.
\end{align*}
If, in addition, $\partial_\omega^2F(\omega_*,a_*,b_*)\neq0$ and $\dim\ker\mathcal L(\omega_*,a_*,b_*)=1$, then, the resulting points remain second-order exceptional points for $b$ sufficiently close to $b_*$.
\end{theorem}

\begin{proof}
The proof of this statement is provided in Appendix~\ref{app:proof-transversality}.
\end{proof}


\section{Second-Order Exceptional Points in Dimers}
\label{sec:dimers}

We now specialize the abstract determinant conditions to the first non-trivial array, namely a dimer. This section makes the general theory explicit: the coalescence conditions become threshold equations involving detuning, damping and coupling, and the nonlinear dependence on $\omega$ remains visible through the entries of the reduced pencil.

\subsection{General dimer criterion}

Let us first consider the case of two coupled resonators. The reduced resonance problem takes the form
\begin{align*}
  \mathcal L(\omega,p)
  =
  \begin{pmatrix}
    a_1(\omega,p) & b(\omega,p) \\
    c(\omega,p) & a_2(\omega,p)
  \end{pmatrix},
\end{align*}
so that
\begin{align*}
  F(\omega,p)
  =
  a_1(\omega,p)a_2(\omega,p)-b(\omega,p)c(\omega,p).
\end{align*}
A second-order exceptional point satisfies
\begin{align*}
  F(\omega_*,p_*)=0,
  \qquad
  \partial_\omega F(\omega_*,p_*)=0,
  \qquad
  \dim\ker \mathcal L(\omega_*,p_*)=1.
\end{align*}
In terms of the matrix entries, the first two conditions are
\begin{align*}
  a_1a_2-bc=0 \qquad \text{and} \qquad  a_1'a_2+a_1a_2'-b'c-bc'=0,
\end{align*}
evaluated at $(\omega_*,p_*)$. Here and throughout this section primes denote differentiation with respect to $\omega$.

\begin{lemma}[Dimer exceptional point criterion]
\label{lem:dimer-criterion}
Assume that the entries of $\mathcal L(\cdot,p_*)$ are holomorphic near
$\omega_*$. If
\begin{align*}
  a_1(\omega_*,p_*)a_2(\omega_*,p_*)
  -
  b(\omega_*,p_*)c(\omega_*,p_*)=0,
\end{align*}
\begin{align*}
  a_1'(\omega_*,p_*)a_2(\omega_*,p_*)
  +
  a_1(\omega_*,p_*)a_2'(\omega_*,p_*)
  -
  b'(\omega_*,p_*)c(\omega_*,p_*)
  -
  b(\omega_*,p_*)c'(\omega_*,p_*)=0,
\end{align*}
\begin{align*}
  \partial_\omega^2F(\omega_*,p_*)\neq 0,
  \qquad
  \dim\ker\mathcal L(\omega_*,p_*)=1,
\end{align*}
then, $(\omega_*,p_*)$ is a second-order exceptional point of the reduced dimer.
\end{lemma}

\begin{proof}
The proof of this statement is provided in Appendix~\ref{app:proof-dimer-criterion}.
\end{proof}

\subsection{Balanced detuning and symmetric coupling}

We now specialize this criterion to a dimer model which separates the common local dispersive response from the detuning and the coupling. We write
\begin{align*}
  \mathcal L(\omega,p)
  =
  \begin{pmatrix}
    A(\omega,p)+\Delta(\omega,p) & -\kappa(\omega,p) \\
    -\kappa(\omega,p) & A(\omega,p)-\Delta(\omega,p)
  \end{pmatrix}.
\end{align*}
Here $A$ represents the common isolated resonator response, $\Delta$ represents detuning or balanced gain/loss asymmetry, and $\kappa$ represents the projected volume-mediated coupling between the two resonators. The determinant is
\begin{align*}
  F(\omega,p)
  =
  A(\omega,p)^2
  -
  \Delta(\omega,p)^2
  -
  \kappa(\omega,p)^2.
\end{align*}
Hence a second-order exceptional point must satisfy
\begin{align*}
  A(\omega_*,p_*)^2
  -
  \Delta(\omega_*,p_*)^2
  -
  \kappa(\omega_*,p_*)^2
  =
  0,
\end{align*}
\begin{align*}
  A(\omega_*,p_*)A'(\omega_*,p_*)
  -
  \Delta(\omega_*,p_*)\Delta'(\omega_*,p_*)
  -
  \kappa(\omega_*,p_*)\kappa'(\omega_*,p_*)
  =
  0.
\end{align*}
The geometric condition is automatic unless the whole matrix vanishes. Indeed, a nonzero singular $2\times2$ matrix has a one-dimensional kernel.

\begin{lemma}[Balanced dimer exceptional point equations]
\label{lem:balanced-dimer}
Assume that $A,\Delta$ and $\kappa$ are holomorphic near $\omega_*$. Suppose
that
\begin{align*}
  A_*^2-\Delta_*^2-\kappa_*^2=0, \qquad A_*A_*'-\Delta_*\Delta_*'-\kappa_*\kappa_*'=0
\end{align*}
and
\begin{align*}
  \partial_\omega^2
  \left(
    A^2-\Delta^2-\kappa^2
  \right)(\omega_*,p_*)\neq 0,
\end{align*}
where $A_*=A(\omega_*,p_*)$, $\Delta_*=\Delta(\omega_*,p_*)$ and $\kappa_*=\kappa(\omega_*,p_*)$. If
\begin{align*}
  \mathcal L(\omega_*,p_*)\neq 0,
\end{align*}
then, $(\omega_*,p_*)$ is a second-order exceptional point.
\end{lemma}

\begin{proof}
The proof of this statement is provided in Appendix~\ref{app:proof-balanced-dimer}.
\end{proof}

\subsection{A constant-threshold model}

The previous lemma allows both the detuning and the coupling to depend on the frequency. A particularly transparent limit is obtained by treating these quantities as constant over the frequency window of interest:
\begin{align*}
  \Delta(\omega,p)=\Delta_0,
  \qquad
  \kappa(\omega,p)=\kappa_0.
\end{align*}
Then,
\begin{align*}
  F(\omega)
  =
  A(\omega)^2-\Delta_0^2-\kappa_0^2.
\end{align*}
If $A'(\omega_*)\neq0$, the exceptional point conditions reduce to
\begin{align*}
  A(\omega_*)=0,
  \qquad
  \Delta_0^2+\kappa_0^2=0.
\end{align*}

\begin{corollary}[PT-type dimer threshold]
\label{cor:pt-threshold}
Consider the balanced dimer with constant coupling $\kappa_0\neq0$ and
constant imaginary detuning
\begin{align*}
  \Delta_0=i\Gamma,
  \qquad
  \Gamma\in\mathbb R.
\end{align*}
Assume that $A$ is holomorphic near $\omega_*$ and that
\begin{align*}
  A(\omega_*)=0,
  \qquad
  A'(\omega_*)\neq0.
\end{align*}
Then, $(\omega_*,\Gamma)$ is a second-order exceptional point if and only if
\begin{align*}
  \Gamma^2=\kappa_0^2.
\end{align*}
\end{corollary}

\begin{proof}
The proof of this statement is provided in Appendix~\ref{app:proof-pt-threshold}.
\end{proof}

\begin{remark}
This corollary recovers the familiar balanced gain/loss threshold. However, the location of the exceptional point is determined by the nonlinear equation $A(\omega_*)=0$. In the present problem, $A$ contains the dispersive material law and the self-interaction contribution coming from the volume-potential reduction. Thus, the threshold relation has the same form as in the standard two-mode model, while the frequency $\omega_*$ and the local splitting constants are controlled by the highly dispersive resonance equation.
\end{remark}

\section{Perturbation and Sensitivity Near an Exceptional Point}
\label{sec:perturbation}

The determinant conditions locate an exceptional point. The next question is what they imply for nearby resonances. This is where the sensing motivation enters the analysis: a defective characteristic value is useful only through the way it unfolds under perturbations, and in the dispersive setting the unfolding constants depend on the material derivatives of the reduced pencil.

\subsection{Local determinant expansion}

We now study what happens when the system is perturbed away from an exceptional point. Let $(\omega_*,p_*)$ be an exceptional point of order $m$, and let
\begin{align*}
  p(\tau)=p_*+\tau h
\end{align*}
be a one-parameter perturbation of the geometry or material parameters. We set
\begin{align*}
  \Phi(\omega,\tau)
  :=
  F(\omega,p(\tau)).
\end{align*}
Since $\omega_*$ is an order-$m$ zero of $F(\cdot,p_*)$, the Taylor expansion of $\Phi$ near $(\omega_*,0)$ begins as
\begin{align*}
  \Phi(\omega,\tau)
  =
  a_m(\omega-\omega_*)^m
  +
  b\tau
  +
  \text{higher-order terms},
\end{align*}
where
\begin{align*}
  a_m
  =
  \frac{1}{m!}\partial_\omega^m F(\omega_*,p_*),
  \qquad
  b
  =
  \left.
  \frac{d}{d\tau}
  F(\omega_*,p_*+\tau h)
  \right|_{\tau=0}.
\end{align*}
The coefficient $a_m$ is nonzero by the exceptional-point condition. The coefficient $b$ measures whether the chosen perturbation unfolds the exceptional point at first order. This is the generic case, and it is the one which gives the usual fractional-power splitting.

\subsection{Puiseux splitting}

Puiseux splitting refers to the local branching of roots of an analytic characteristic equation in fractional powers of a perturbation parameter. This is the natural expansion near a defective characteristic value: instead of being analytic functions of the perturbation, the resonance branches are analytic in a fractional power of it. This is the finite-dimensional form of the analytic perturbation theory for operator and matrix pencils \cite{kato1966perturbation}.

\begin{theorem}[Generic splitting near an exceptional point]
\label{thm:generic-puiseux}
Let $(\omega_*,p_*)$ be an exceptional point of order $m$, let
$p(\tau)=p_*+\tau h$, and set
\begin{align*}
  \Phi(\omega,\tau)
  :=
  F(\omega,p(\tau)).
\end{align*}
Assume that $\Phi$ is holomorphic in $(\omega,\tau)$ near $(\omega_*,0)$ and that
\begin{align*}
  a_m
  :=
  \frac{1}{m!}\partial_\omega^mF(\omega_*,p_*)
  \neq0.
\end{align*}
Assume, moreover, that the perturbation is generic in the sense that
\begin{align*}
  b
  :=
  \left.
  \frac{d}{d\tau}
  F(\omega_*,p_*+\tau h)
  \right|_{\tau=0}
  \neq0.
\end{align*}
Let $c_1,\ldots,c_m$ be the $m$ roots of
\begin{align*}
  a_m c^m+b=0.
\end{align*}
Then, for sufficiently small nonzero $\tau$, the $m$ resonances near $\omega_*$ admit Puiseux expansions
\begin{align*}
  \omega_j(\tau)
  =
  \omega_*
  +
  c_j\tau^{1/m}
  +
  O(\tau^{2/m}),
  \qquad
  j=1,\ldots,m,
\end{align*}
after choosing a branch of $\tau^{1/m}$.
\end{theorem}

\begin{proof}
The proof of this statement is provided in Appendix~\ref{app:proof-puiseux}.
\end{proof}

\begin{corollary}[Square-root splitting for a dimer]
\label{cor:square-root-dimer}
Let $(\omega_*,p_*)$ be a second-order exceptional point of the reduced dimer,
and assume that the perturbation direction $h$ satisfies
\begin{align*}
  b
  =
  \left.
  \frac{d}{d\tau}
  F(\omega_*,p_*+\tau h)
  \right|_{\tau=0}
  \neq0.
\end{align*}
Then, the two resonances near $\omega_*$ satisfy
\begin{align*}
  \omega_\pm(\tau)
  =
  \omega_*
  \pm
  \left(
    -
    \frac{2b}{\partial_\omega^2F(\omega_*,p_*)}
  \right)^{1/2}
  \tau^{1/2}
  +
  O(\tau).
\end{align*}
\end{corollary}

\begin{proof}
The proof of this statement is provided in Appendix~\ref{app:proof-square-root}.
\end{proof}

\subsection{Sensitivity prefactor}

The coefficients in the Puiseux expansion can be written directly in terms of the reduced matrix. This makes the role of the material derivatives explicit.

\begin{lemma}[Puiseux prefactor and material derivatives]
\label{lem:prefactor-material-derivatives}
Let $(\omega_*,p_*)$ be a second-order exceptional point of the reduced
matrix, and let $p(\tau)=p_*+\tau h$ be a generic perturbation. Then,
\begin{align*}
  \omega_\pm(\tau)
  =
  \omega_*
  \pm
  \mathfrak c\,\tau^{1/2}
  +
  O(\tau),
\end{align*}
where
\begin{align*}
  \mathfrak c^2
  =
  -
  \frac{
    2\tr\bigl(
      \adj\mathcal L(\omega_*,p_*)
      D_p\mathcal L(\omega_*,p_*)[h]
    \bigr)
  }{
    \partial_\omega^2F(\omega_*,p_*)
  }.
\end{align*}
Moreover,
\begin{align*}
  \partial_\omega F(\omega,p)
  =
  \tr\bigl(
    \adj\mathcal L(\omega,p)
    \partial_\omega\mathcal L(\omega,p)
  \bigr),
\end{align*}
so the denominator $\partial_\omega^2F(\omega_*,p_*)$ contains the frequency derivatives of the material contrasts, including terms involving $\partial_\omega\xi_i$ and $\partial_\omega^2\xi_i$ whenever these appear in the entries of $\mathcal L$.
\end{lemma}

\begin{proof}
The proof of this statement is provided in Appendix~\ref{app:proof-prefactor-material}.
\end{proof}

The splitting law gives the sensitivity estimate
\begin{align*}
  \max_{j,k}
  |\omega_j(\tau)-\omega_k(\tau)|
  =
  O(|\tau|^{1/m}).
\end{align*}
For a generic perturbation the leading coefficient is nonzero, so the scaling is of order $|\tau|^{1/m}$. More explicitly, if $\{c_1,\ldots,c_m\}$ are the $m$ roots of $a_m c^m+b=0$, then,
\begin{align*}
  \max_{j,k}
  |\omega_j(\tau)-\omega_k(\tau)|
  =
  \left(
    \max_{j,k}|c_j-c_k|
  \right)
  |\tau|^{1/m}
  +
  O(|\tau|^{2/m}).
\end{align*}

\begin{remark}
The exponent $1/m$ is the usual exceptional-point sensitivity exponent. The new feature here is the structure of the prefactor. Both $a_m$ and $b$ are derivatives of the nonlinear determinant $F(\omega,p)=\det\mathcal L(\omega,p)$. Hence they contain derivatives of the dispersive material law, the local self-interaction terms and the volume-mediated coupling coefficients. Highly dispersive resonators therefore do not only reproduce the usual exceptional-point splitting law; they also provide material and geometric parameters which can tune its prefactor.
\end{remark}

\begin{remark}
Theorem~\ref{thm:generic-puiseux} is stated for the leading reduced determinant. For the full projected determinant $F_\delta=\det(\mathcal L+\mathcal E_\delta)$, the same expansion describes the leading behaviour whenever the perturbation scale is resolved above the reduction error. In particular, the resonance cluster of Corollary~\ref{cor:exceptional-cluster} is shifted by the finite-size error, whereas the fractional exponent is the local exponent of the limiting exceptional point.
\end{remark}

\section{Higher-Order Systems}
\label{sec:higher-order}

The same mechanism is not restricted to dimers. Once exceptional points are formulated through the order of a zero of the nonlinear determinant and the dimension of the nullspace, higher-order degeneracies can be treated by the same principle. The array supplies the material and geometric parameters needed to impose the higher-order conditions.

\subsection{Higher-order degeneracy equations}

For systems with three or more resonators, the same determinant conditions \eqref{eq:ep-derivative-conditions}--\eqref{eq:ep-nondegeneracy-conditions} give a direct way to formulate higher-order exceptional points. Equivalently, one solves the nonlinear degeneracy system
\begin{align*}
  \mathcal E_m(\omega,p)
  :=
  \left(
    F(\omega,p),
    \partial_\omega F(\omega,p),
    \ldots,
    \partial_\omega^{m-1}F(\omega,p)
  \right)
  =
  0,
\end{align*}
and subsequently verifies the non-degeneracy and geometric conditions in \eqref{eq:ep-nondegeneracy-conditions}.

\begin{lemma}[Higher-order reduced exceptional point criterion]
\label{lem:higher-order-criterion}
Assume that $\mathcal L(\cdot,p_*)$ is holomorphic near $\omega_*$ and that
$F(\cdot,p_*)$ has a zero of order $m$ at $\omega_*$. If
\begin{align*}
  \dim\ker\mathcal L(\omega_*,p_*)=1,
\end{align*}
then, $(\omega_*,p_*)$ is an exceptional point of order $m$ of the reduced $N$-resonator system.
\end{lemma}

\begin{proof}
The proof of this statement is provided in Appendix~\ref{app:proof-higher-order}.
\end{proof}

\subsection{Trimers and third-order exceptional points}

The first genuinely higher-order case is a trimer. A third-order exceptional point is characterized by
\begin{align*}
  F(\omega_*,p_*)=0,
  \qquad
  \partial_\omega F(\omega_*,p_*)=0,
  \qquad
  \partial_\omega^2F(\omega_*,p_*)=0,
\end{align*}
with
\begin{align*}
  \partial_\omega^3F(\omega_*,p_*)\neq0,
  \qquad
  \dim\ker\mathcal L_3(\omega_*,p_*)=1.
\end{align*}
Here $\mathcal L_3$ denotes the $3\times3$ reduced matrix. A convenient abstract form is
\begin{align*}
  \mathcal L_3(\omega,p)
  =
  A(\omega,p)I-H(\omega,p),
\end{align*}
where $A$ is a common dispersive local response and $H$ is an effective coupling and detuning matrix. In this representation, a third-order exceptional point corresponds to $A(\omega_*,p_*)$ being a defective eigenvalue of $H(\omega_*,p_*)$ with algebraic multiplicity three and geometric multiplicity one, together with the nonlinear frequency degeneracy encoded by the determinant conditions above.

The following normal-form result gives an explicit third-order exceptional point for a nonlinear dispersive trimer pencil.

\begin{lemma}[Third-order trimer exceptional point]
\label{lem:trimer-third-order}
Let $A$ be holomorphic near $\omega_*$ and assume
\begin{align*}
  A(\omega_*)=0,
  \qquad
  A'(\omega_*)\neq0.
\end{align*}
Let $\kappa\neq0$ and consider the trimer pencil
\begin{align*}
  \mathcal L_3(\omega,\tau)
  =
  \begin{pmatrix}
    A(\omega) & -\kappa & 0\\
    0 & A(\omega) & -\kappa\\
    -\tau & 0 & A(\omega)
  \end{pmatrix}.
\end{align*}
Then, $(\omega_*,0)$ is a third-order exceptional point of the reduced trimer system.
\end{lemma}

\begin{proof}
The proof of this statement is provided in Appendix~\ref{app:proof-trimer}.
\end{proof}

For $\tau\neq0$, the same normal form gives the local unfolding explicitly:
\begin{align*}
  A(\omega)^3=\kappa^2\tau.
\end{align*}
Thus, the three nearby resonances satisfy a cube-root splitting law $|\omega_j(\tau)-\omega_*|=O(|\tau|^{1/3})$, in agreement with Theorem~\ref{thm:generic-puiseux} for $m=3$.

\subsection{Volume-interaction arrays}

The volume-interaction form of the reduced matrix supplies a natural set of parameters for higher-order construction. In the constant-mode approximation, one may write schematically
\begin{align*}
  \mathcal L(\omega,p)
  =
  I-\operatorname{diag}
  \bigl(
    B_1(\omega,p),\ldots,B_N(\omega,p)
  \bigr)
  C(\omega,p),
\end{align*}
where the factors $B_i$ encode the local dispersive response and $C$ contains the projected volume-mediated interactions. Material tuning primarily changes the diagonal factors $B_i$, while geometric tuning changes both the interaction matrix $C$ and, through the component shapes and sizes, the local resonant factors.

This structure is useful for higher-order exceptional points. In non-dispersive matrix models, the number of available parameters is often limited to coupling strengths and gain/loss coefficients. Here the nonlinear dispersive law creates additional frequency-dependent degrees of freedom through $B_i(\omega,p)$. Thus, higher-order degeneracies can be sought by combining material detuning, damping, gain/loss and volume-mediated coupling.

\begin{remark}
The higher-order theory is formulated at the level of the reduced nonlinear matrix. The passage from the full scattering problem to this matrix is supplied by the subwavelength volume-potential reduction. The numerical experiments below therefore focus on locating and testing higher-order exceptional points of $\mathcal L(\omega,p)$, together with the Puiseux splitting predicted by the perturbation theory.
\end{remark}

\section{Numerical Experiments}
\label{sec:numerics}

Let us now test the reduced theory at the level of the nonlinear matrix pencil. The numerical examples are organized to follow the logic of the paper: the determinant criterion for dimers, passive generation without imposed gain, a higher-order trimer, inverse design, and finally explicit halide-perovskite-inspired material laws.

The computations solve the reduced determinant equations introduced in Sections~\ref{sec:nonlinear-pencils} and~\ref{sec:higher-order}. For each candidate exceptional point, we check the relevant non-degeneracy condition and the dimension of the numerical nullspace. The numerical parameters, residuals, fitted exponents and output files used to generate the figures are recorded in Appendix~\ref{app:numerical-reproducibility}.

\subsection{Frequency-dependent volume-interaction dimer}

We start by considering a balanced dimer with frequency-dependent volume-mediated coupling,
\begin{align*}
  \mathcal L(\omega,\Gamma)
  =
  \begin{pmatrix}
    A(\omega)+i\Gamma & -\kappa(\omega)\\
    -\kappa(\omega) & A(\omega)-i\Gamma
  \end{pmatrix}.
\end{align*}
Here $A$ is a nonlinear dispersive local response and $\kappa(\omega)$ contains the projected Helmholtz Green interaction. The parameter $\Gamma$ is tuned so that the determinant has a double zero. This example is the numerical counterpart of the dimer criterion in Section~\ref{sec:dimers}, with the nonlinear frequency dependence described in Lemma~\ref{lem:dispersion-nonlinear-pencil}.

\begin{figure}[htbp]
  \centering
  \begin{subfigure}[t]{0.32\textwidth}
    \centering
    \includegraphics[width=\textwidth]{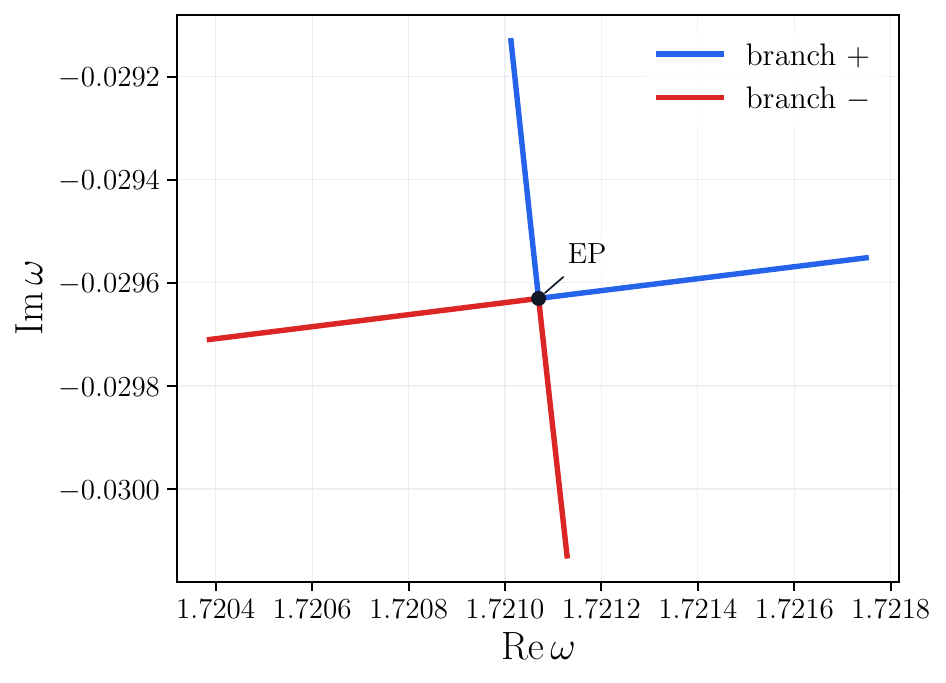}
    \caption{Resonance trajectories.}
    \label{fig:volume-dimer-trajectories}
  \end{subfigure}
  \hfill
  \begin{subfigure}[t]{0.32\textwidth}
    \centering
    \includegraphics[width=\textwidth]{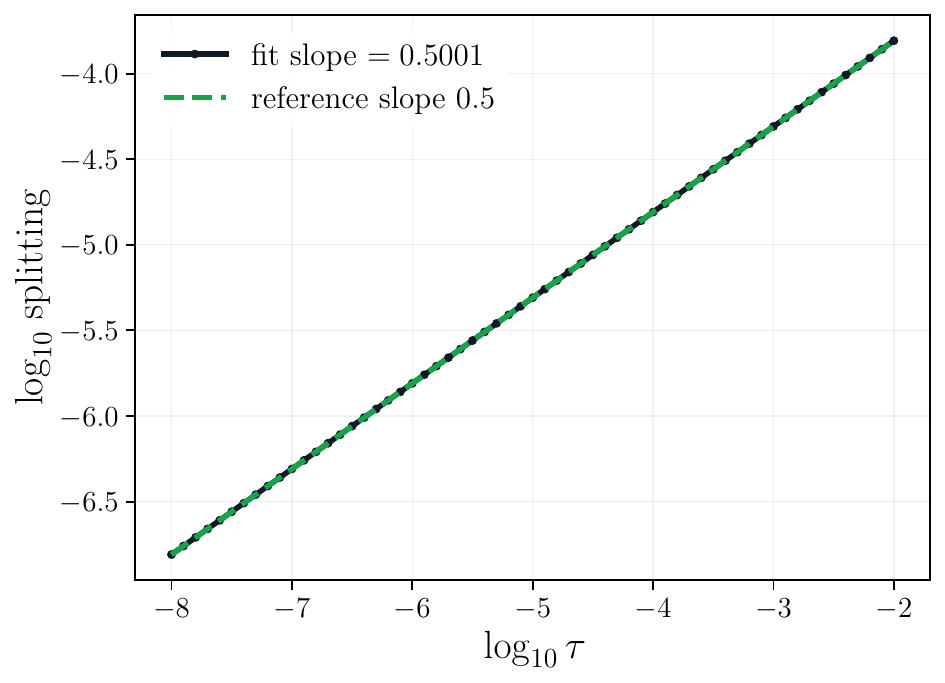}
    \caption{Square-root splitting.}
    \label{fig:volume-dimer-scaling}
  \end{subfigure}
  \hfill
  \begin{subfigure}[t]{0.32\textwidth}
    \centering
    \includegraphics[width=\textwidth]{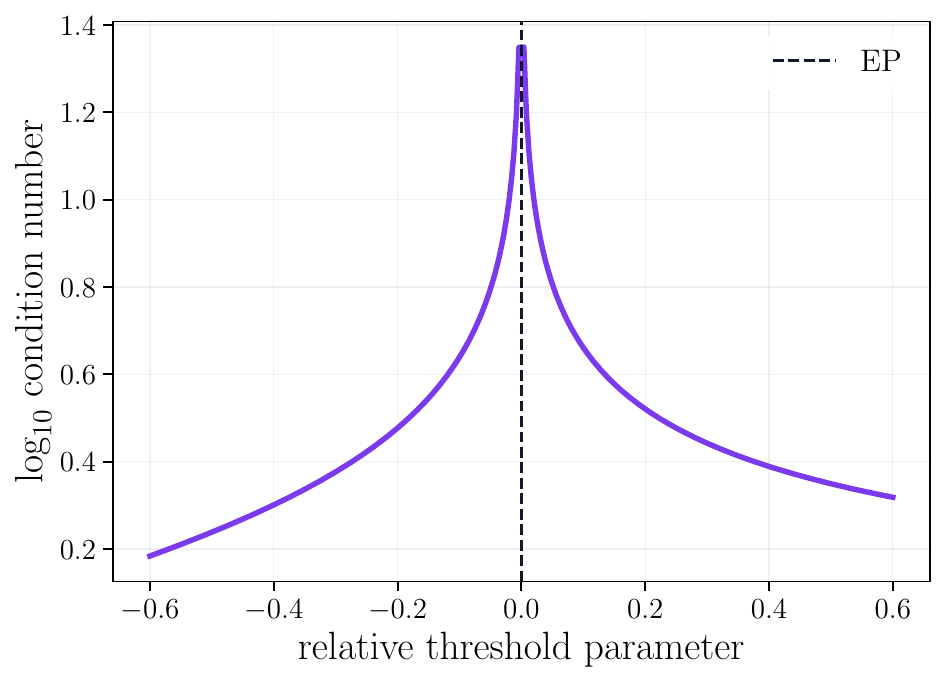}
    \caption{Modal condition number.}
    \label{fig:volume-dimer-condition}
  \end{subfigure}
  \caption{Frequency-dependent volume-interaction dimer. The panels show the resonance trajectories, the local splitting law and the modal condition number near the exceptional point.}
  \label{fig:volume-dimer}
\end{figure}

Figure~\ref{fig:volume-dimer} shows the three diagnostics used throughout the numerical section. The two resonance branches meet at the computed exceptional point, and the associated reduced eigenvectors become nearly parallel, as shown by the growth of the modal condition number. The log-log splitting plot has fitted slope close to $1/2$, which is consistent with Corollary~\ref{cor:square-root-dimer}. Thus, the computation illustrates both parts of the second-order criterion: the determinant has a double zero and the geometric multiplicity is one.

\subsection{Passive exceptional point without gain}

The previous example uses a balanced gain/loss parameter. We next give an example in which an exceptional point is generated by passive material damping and material detuning. We use the reduced matrix
\begin{align*}
  \mathcal L(\omega)
  =
  \begin{pmatrix}
    A_1(\omega) & -\kappa(\omega)\\
    -\kappa(\omega) & A_2(\omega)
  \end{pmatrix},
\end{align*}
where both resonators have positive damping coefficients. With $\gamma_1>0$, we tune the second resonator by changing its material parameters. The computed exceptional point has $\gamma_2>0$, so the non-Hermiticity is produced by loss and detuning, rather than by explicit gain.

\begin{figure}[htbp]
  \centering
  \begin{subfigure}[t]{0.32\textwidth}
    \centering
    \includegraphics[width=\textwidth]{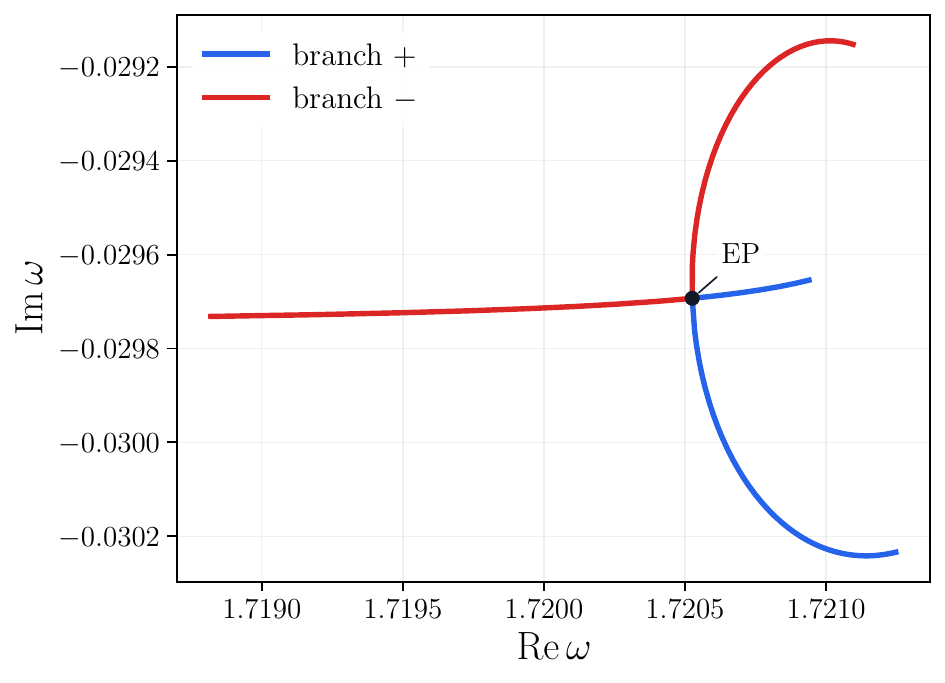}
    \caption{Resonance trajectories.}
    \label{fig:passive-dimer-trajectories}
  \end{subfigure}
  \hfill
  \begin{subfigure}[t]{0.32\textwidth}
    \centering
    \includegraphics[width=\textwidth]{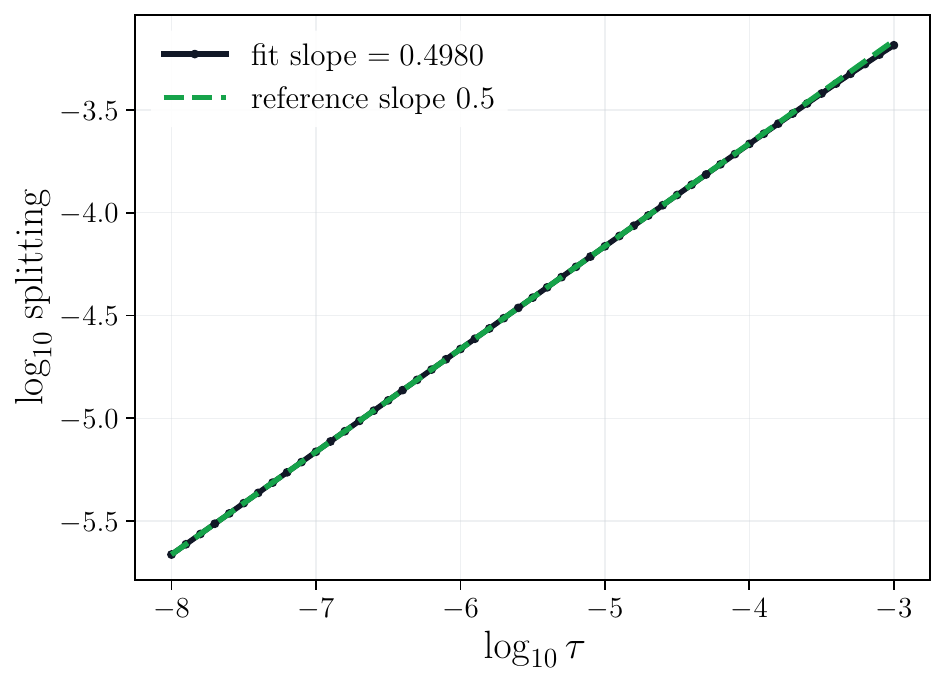}
    \caption{Square-root splitting.}
    \label{fig:passive-dimer-scaling}
  \end{subfigure}
  \hfill
  \begin{subfigure}[t]{0.32\textwidth}
    \centering
    \includegraphics[width=\textwidth]{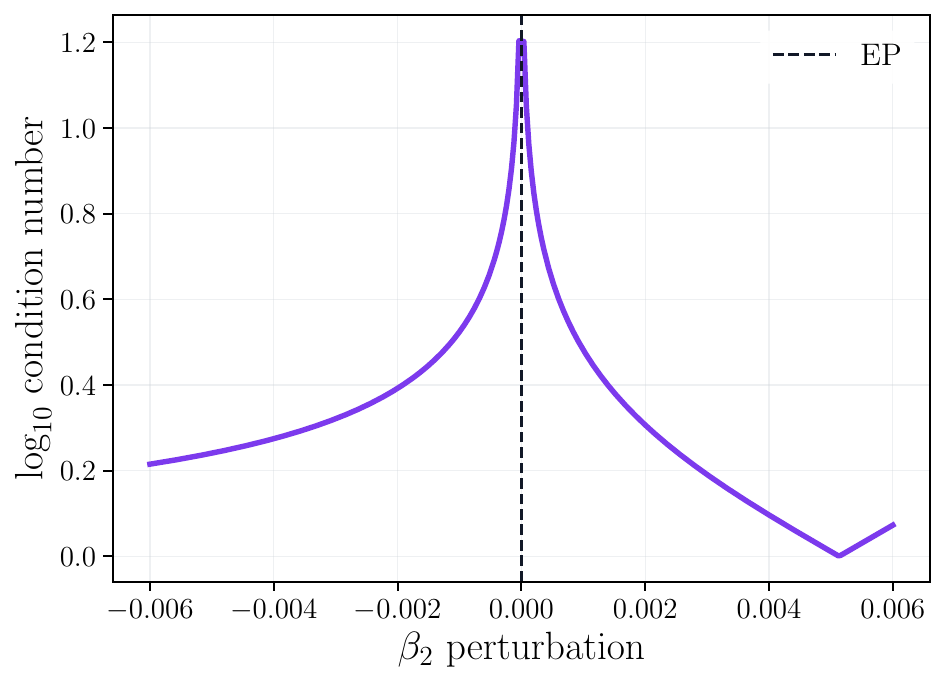}
    \caption{Modal condition number.}
    \label{fig:passive-dimer-condition}
  \end{subfigure}
  \caption{Passive damped dimer without gain. The panels show coalescence of the two resonance branches, square-root splitting under detuning and growth of the modal condition number near the exceptional point.}
  \label{fig:passive-dimer}
\end{figure}

Figure~\ref{fig:passive-dimer} shows that the same second-order behaviour can be obtained without imposing a balanced gain/loss parameter. The resonance branches coalesce, the modal basis becomes ill-conditioned near the coalescence, and the measured splitting exponent is again close to $1/2$. This agrees with the determinant criterion of Lemma~\ref{lem:dimer-criterion} and shows, at the level of the reduced model, that passive damping and material detuning can provide the required non-Hermiticity.

\subsection{Third-order exceptional point in a trimer}

We next test the higher-order construction of Lemma~\ref{lem:trimer-third-order}. We use the reduced trimer pencil
\begin{align*}
  \mathcal L_3(\omega,\tau)
  =
  \begin{pmatrix}
    A(\omega) & -\kappa & 0\\
    0 & A(\omega) & -\kappa\\
    -\tau & 0 & A(\omega)
  \end{pmatrix},
\end{align*}
with the same dispersive scalar response $A$ used in the dimer tests. The determinant satisfies
\begin{align*}
  F(\omega,\tau)=A(\omega)^3-\kappa^2\tau.
\end{align*}
Thus, at $\tau=0$, the reduced matrix has the third-order exceptional point described in Lemma~\ref{lem:trimer-third-order}.

\begin{figure}[htbp]
  \centering
  \begin{subfigure}[t]{0.32\textwidth}
    \centering
    \includegraphics[width=\textwidth]{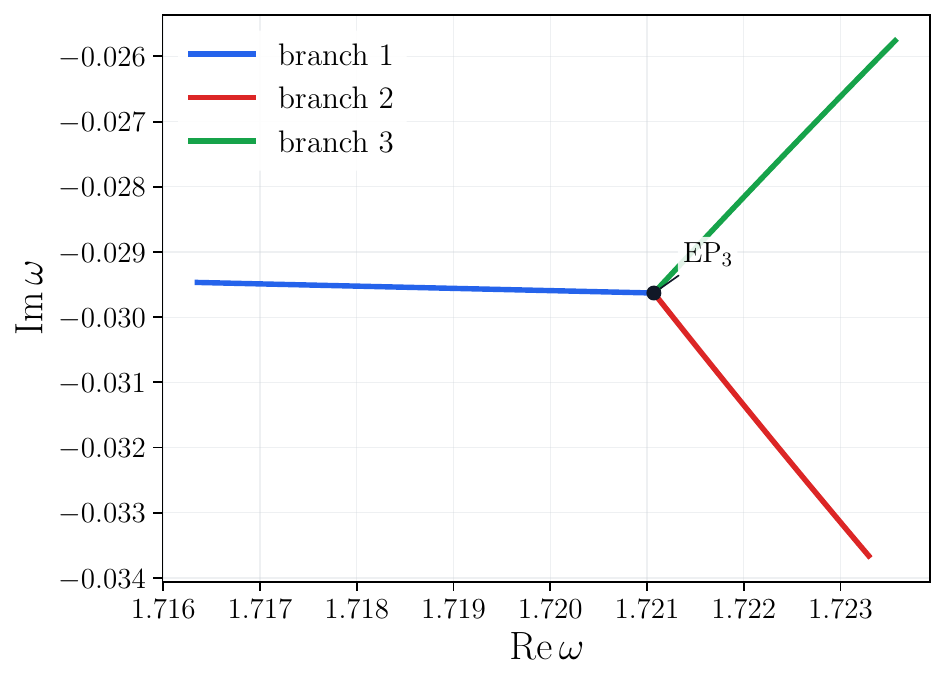}
    \caption{Three resonance branches.}
    \label{fig:trimer-trajectories}
  \end{subfigure}
  \hfill
  \begin{subfigure}[t]{0.32\textwidth}
    \centering
    \includegraphics[width=\textwidth]{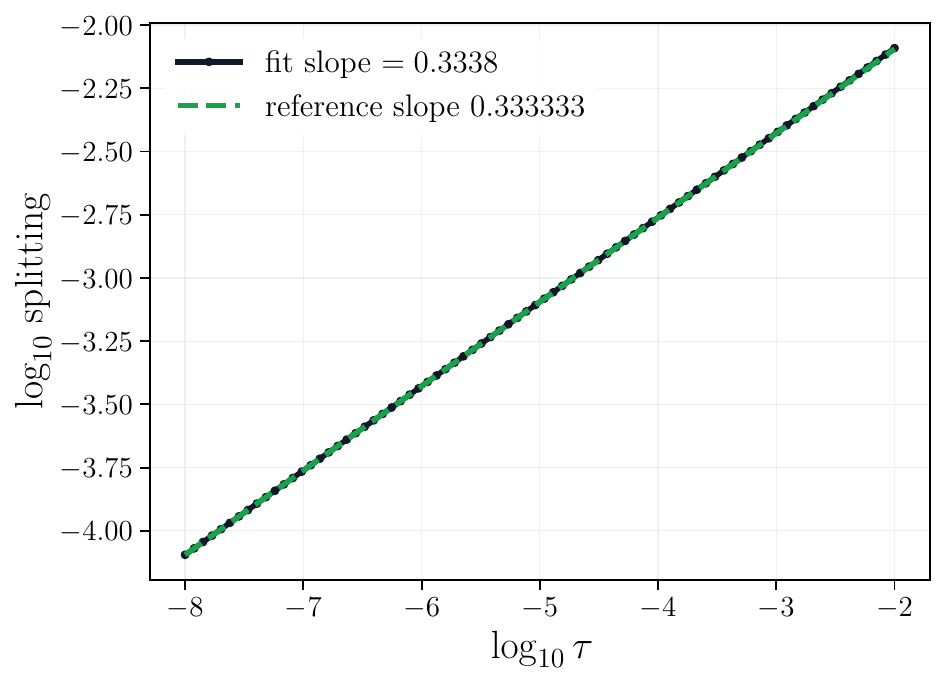}
    \caption{Cube-root splitting.}
    \label{fig:trimer-scaling}
  \end{subfigure}
  \hfill
  \begin{subfigure}[t]{0.32\textwidth}
    \centering
    \includegraphics[width=\textwidth]{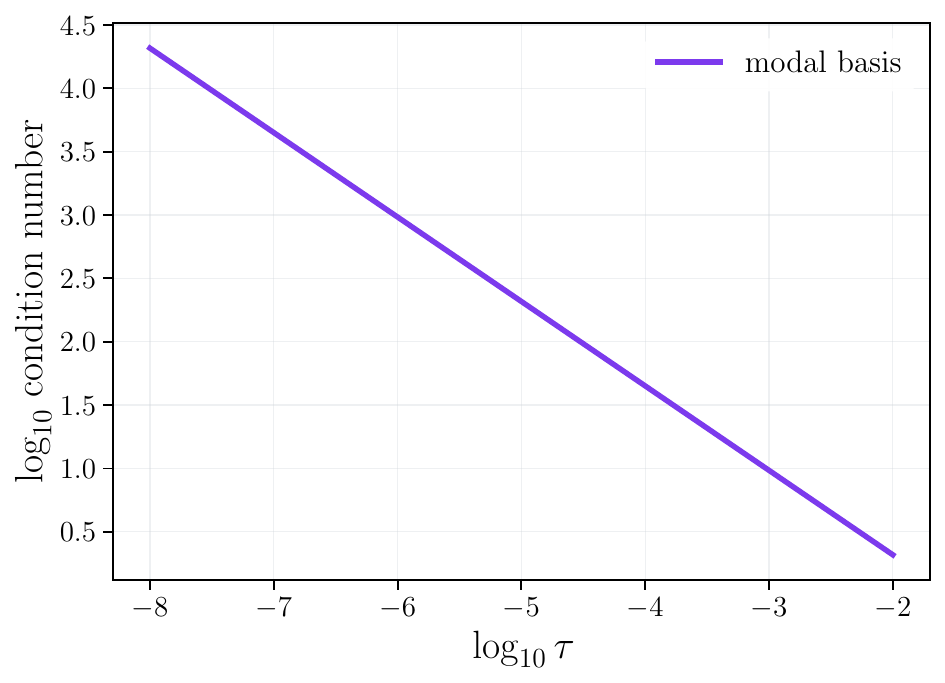}
    \caption{Modal condition number.}
    \label{fig:trimer-condition}
  \end{subfigure}
  \caption{Third-order exceptional point in the reduced trimer. The panels show the coalescence of three resonance branches, the cube-root splitting law and the growth of the modal condition number.}
  \label{fig:trimer-validation}
\end{figure}

Figure~\ref{fig:trimer-validation} shows the higher-order analogue of the dimer computations. The three branches meet at the unperturbed point and the modal basis becomes increasingly ill-conditioned as the perturbation tends to zero. The fitted exponent is close to $1/3$, in agreement with Theorem~\ref{thm:generic-puiseux} for an exceptional point of order three. This provides a numerical check of the trimer normal form from Lemma~\ref{lem:trimer-third-order}.

\subsection{Inverse design of a prescribed exceptional point}

The exceptional-point equations can also be used constructively. We prescribe a target complex frequency and solve the real and imaginary parts of $F=\partial_\omega F=0$ for material and geometric parameters. This turns the determinant criterion into an inverse-design rule for producing a reduced exceptional point at a chosen location. The computation is the numerical counterpart of Theorem~\ref{thm:transversality-inverse-design}.

\begin{figure}[htbp]
  \centering
  \begin{subfigure}[t]{0.48\textwidth}
    \centering
    \includegraphics[width=\textwidth]{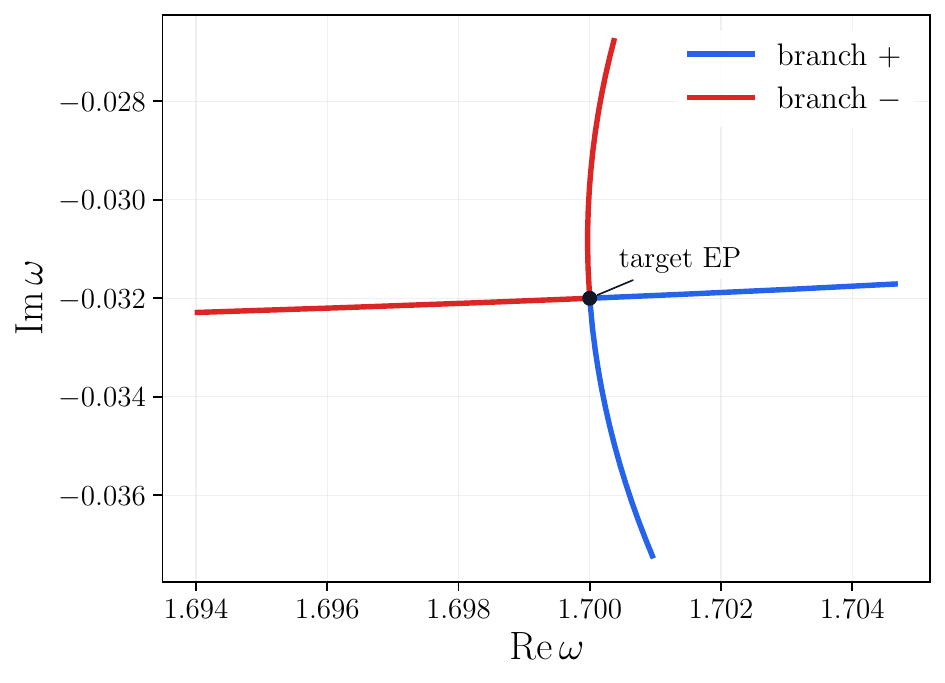}
    \caption{Targeted coalescence.}
    \label{fig:inverse-design-target}
  \end{subfigure}
  \hfill
  \begin{subfigure}[t]{0.48\textwidth}
    \centering
    \includegraphics[width=\textwidth]{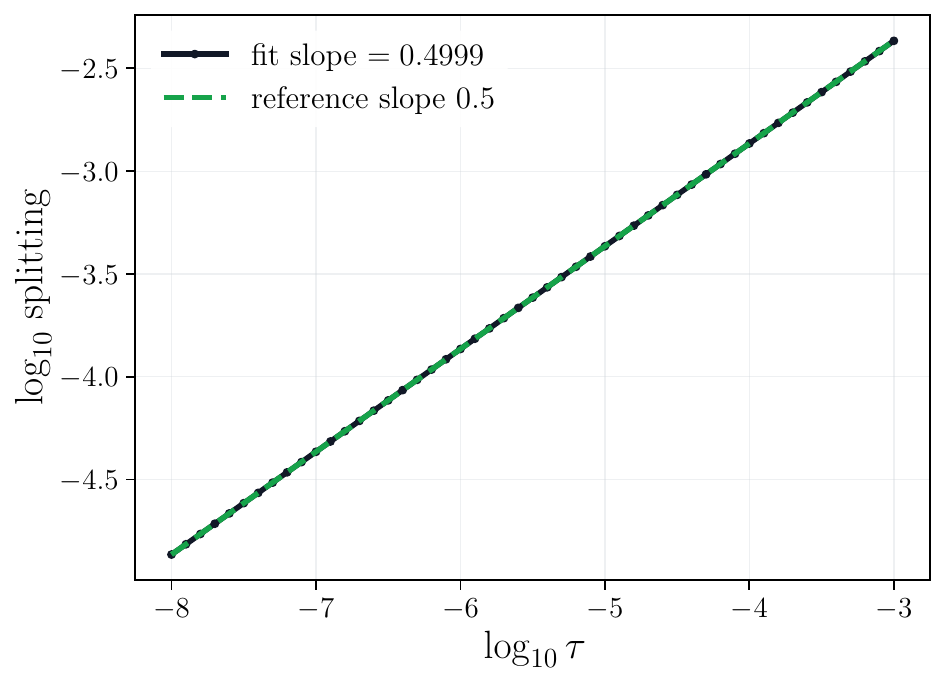}
    \caption{Designed splitting law.}
    \label{fig:inverse-design-scaling}
  \end{subfigure}
  \caption{Inverse design of an exceptional point at a prescribed target frequency. The panels show the targeted coalescence and the square-root splitting law of the designed reduced dimer.}
  \label{fig:inverse-design}
\end{figure}

Figure~\ref{fig:inverse-design} shows that the designed resonance branches coalesce at the prescribed target frequency. When one material parameter is subsequently perturbed, the splitting exponent is close to $1/2$. This is consistent with Corollary~\ref{cor:square-root-dimer} and illustrates that the equations $F=\partial_\omega F=0$ can be used not only to detect exceptional points, but also to design them within the reduced model. At the designed point, the smallest singular value of the real inverse-design Jacobian appearing in Theorem~\ref{thm:transversality-inverse-design} is approximately $4.05\cdot 10^{-3}$, which numerically verifies the transversality condition.

\subsection{Highly dispersive halide-perovskite resonators}

We finish the numerical section with a material-specific reduced model. A highly dispersive resonator is characterized by a permittivity with strong frequency dependence near a material pole. In particular, damping is represented by the imaginary part of the material denominator. Motivated by the optical response of halide perovskites \cite{sutherland2016perovskite,chouhan2020synthesis}, we consider the representative law
\begin{align}
  \varepsilon_i(\omega,k_i)
  =
  \varepsilon_0
  +
  \frac{\alpha_i}{\beta_i-\omega^2+\eta_i k_i^2-i\gamma_i\omega},
  \label{eq:hp-permittivity}
\end{align}
where $\alpha_i,\beta_i,\eta_i,\gamma_i\in\mathbb R$ and $\gamma_i>0$ corresponds to material loss. The pole structure of this function is the concrete source of the nonlinear frequency dependence in the reduced matrix, in the sense of Lemma~\ref{lem:dispersion-nonlinear-pencil}.

For the dimer experiment, this permittivity is inserted into the projected two-resonator matrix through the contrast $\xi_i(\omega,k_i)=\mu_0(\varepsilon_i(\omega,k_i)-\varepsilon_0)$. The second resonator is tuned through its material parameters, while both damping coefficients remain positive. Thus, the example uses a genuine highly dispersive passive material law rather than an imposed gain/loss parameter.

\begin{figure}[htbp]
  \centering
  \begin{subfigure}[t]{0.32\textwidth}
    \centering
    \includegraphics[width=\textwidth]{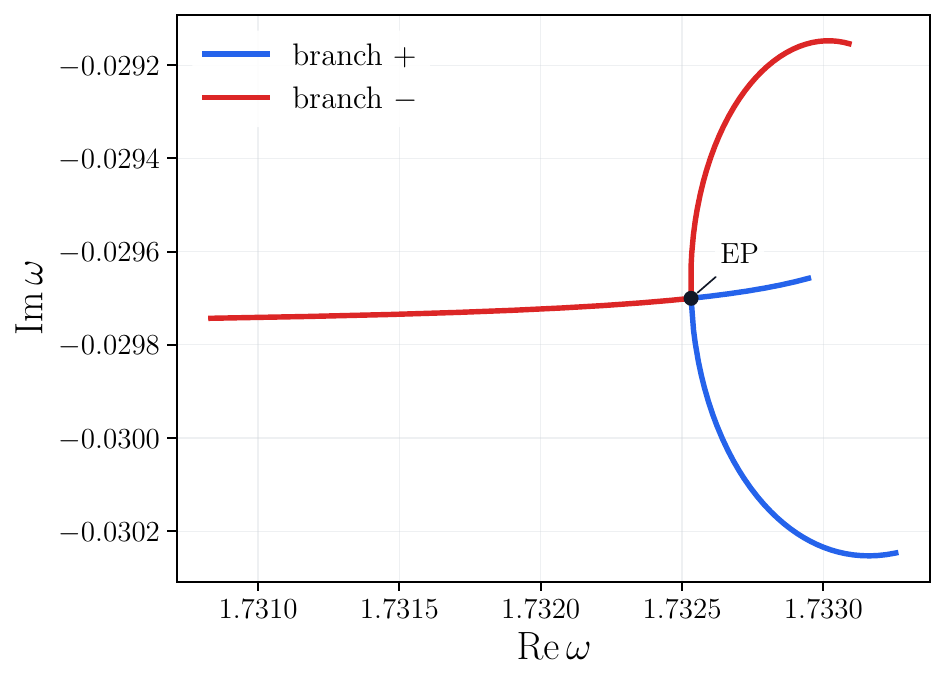}
    \caption{Resonance trajectories.}
    \label{fig:halide-perovskite-trajectories}
  \end{subfigure}
  \hfill
  \begin{subfigure}[t]{0.32\textwidth}
    \centering
    \includegraphics[width=\textwidth]{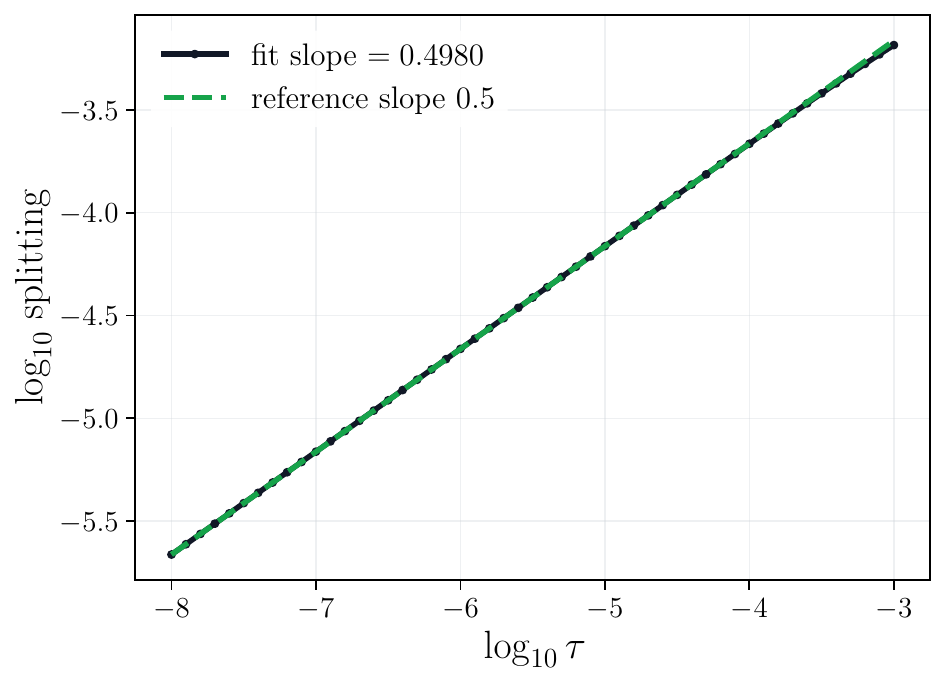}
    \caption{Square-root splitting.}
    \label{fig:halide-perovskite-scaling}
  \end{subfigure}
  \hfill
  \begin{subfigure}[t]{0.32\textwidth}
    \centering
    \includegraphics[width=\textwidth]{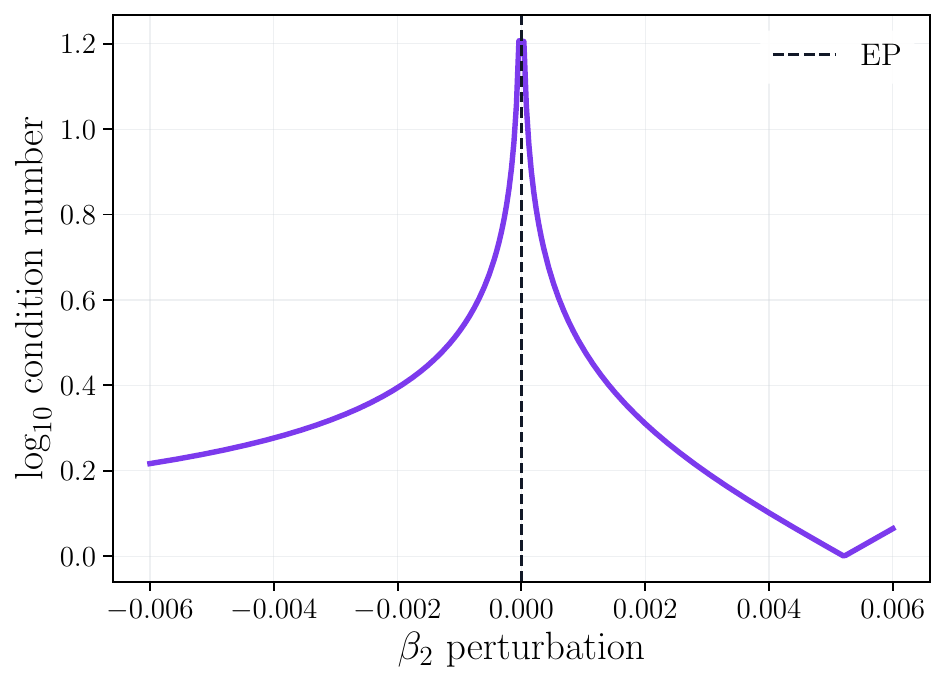}
    \caption{Modal condition number.}
    \label{fig:halide-perovskite-condition}
  \end{subfigure}
  \caption{Halide-perovskite-inspired highly dispersive passive dimer. The panels show coalescence of the resonance branches, square-root splitting under material detuning and growth of the modal condition number near the exceptional point.}
  \label{fig:halide-perovskite-dimer}
\end{figure}

Figure~\ref{fig:halide-perovskite-dimer} shows that the material-specific permittivity law produces the same exceptional-point signatures predicted by the reduced theory. The two branches coalesce at the computed point, the modal condition number grows as the perturbation approaches the exceptional point, and the measured splitting exponent is close to $1/2$. This confirms, in a model using the explicit highly dispersive permittivity, the determinant criterion and the Puiseux law of Corollary~\ref{cor:square-root-dimer}.

We also test the same material law in a full-interaction trimer. All three diagonal entries are generated by the halide-perovskite-inspired local responses, and all pairwise volume-mediated interactions are retained through the projected Green-function coupling. Thus, the trimer is not reduced to a nearest-neighbour normal form. The parameters of the second and third resonators are tuned so that the third-order determinant conditions in Section~\ref{sec:higher-order} hold, while the damping coefficients remain positive.

\begin{figure}[htbp]
  \centering
  \begin{subfigure}[t]{0.32\textwidth}
    \centering
    \includegraphics[width=\textwidth]{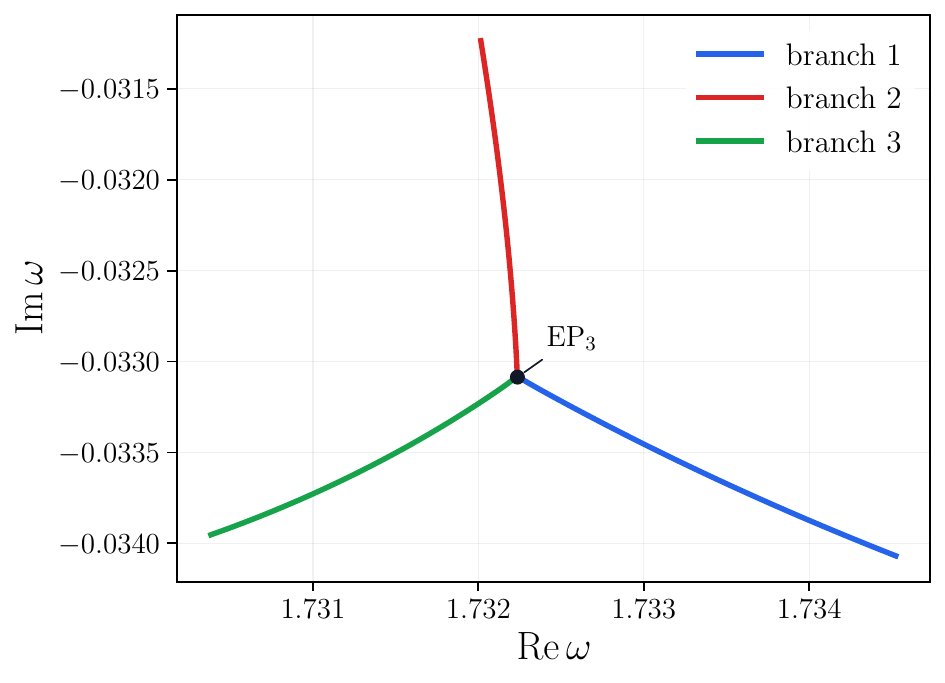}
    \caption{Three resonance branches.}
    \label{fig:halide-perovskite-trimer-trajectories}
  \end{subfigure}
  \hfill
  \begin{subfigure}[t]{0.32\textwidth}
    \centering
    \includegraphics[width=\textwidth]{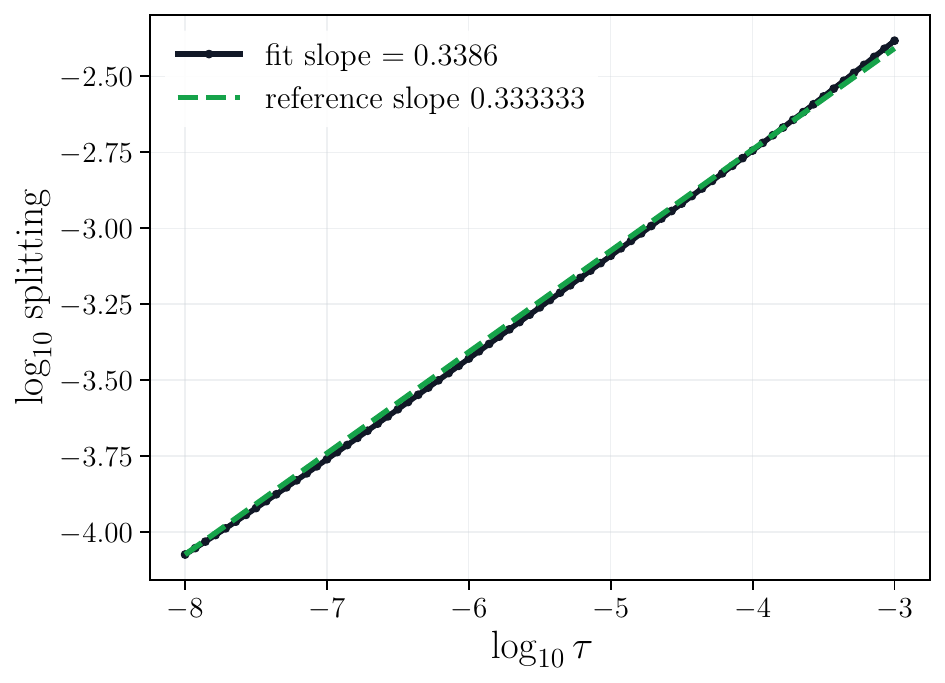}
    \caption{Cube-root splitting.}
    \label{fig:halide-perovskite-trimer-scaling}
  \end{subfigure}
  \hfill
  \begin{subfigure}[t]{0.32\textwidth}
    \centering
    \includegraphics[width=\textwidth]{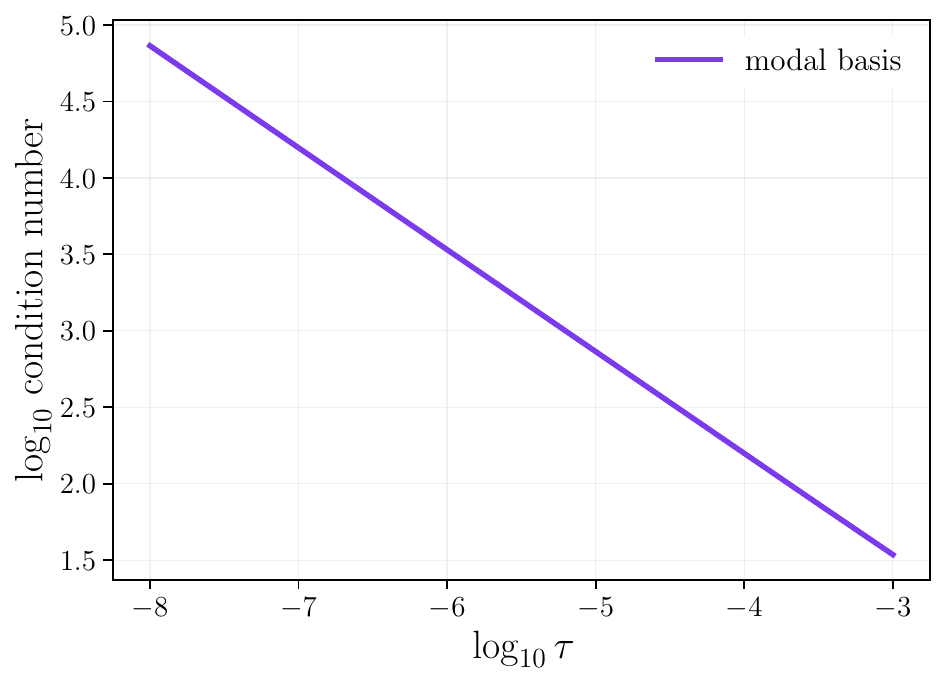}
    \caption{Modal condition number.}
    \label{fig:halide-perovskite-trimer-condition}
  \end{subfigure}
  \caption{Halide-perovskite-inspired full-interaction trimer. The panels show coalescence of three resonance branches, cube-root splitting under material detuning and growth of the modal condition number near the third-order exceptional point.}
  \label{fig:halide-perovskite-trimer}
\end{figure}

Figure~\ref{fig:halide-perovskite-trimer} shows that the higher-order exceptional-point mechanism persists when the explicit highly dispersive permittivity is used in a three-resonator full-interaction model. The three branches coalesce at the computed point and the modal basis becomes ill-conditioned as the perturbation tends to zero. The fitted exponent is close to $1/3$, in agreement with Theorem~\ref{thm:generic-puiseux} and the higher-order determinant conditions of Lemma~\ref{lem:higher-order-criterion}. Thus, the final numerical example shows that the second- and third-order exceptional-point mechanisms remain visible in a genuinely highly dispersive passive material setting.

\section{Conclusion}
\label{sec:conclusion}

In this work, we have studied exceptional points in highly dispersive subwavelength resonator systems. Starting from the Helmholtz problem, we recalled the Lippmann--Schwinger formulation and used the subwavelength reduction to obtain a nonlinear matrix pencil $\mathcal L(\omega,p)$. In this formulation, resonances are characteristic values of the pencil, and second-order exceptional points are detected by the conditions
\begin{align*}
  F(\omega,p)=0,
  \qquad
  \partial_\omega F(\omega,p)=0,
  \qquad
  F(\omega,p)=\det\mathcal L(\omega,p),
\end{align*}
together with the geometric condition that the kernel of $\mathcal L(\omega,p)$ is one-dimensional. This gives a direct connection between the original dispersive scattering problem and the finite-dimensional exceptional-point criteria.

For dimers, we derived explicit threshold conditions describing how a second-order exceptional point is produced by the balance between local detuning, damping and coupling. The balanced dimer recovers the familiar gain/loss threshold as a special case, while the passive examples indicate that this is not the only possible mechanism. In the reduced dispersive resonator systems considered here, non-Hermiticity can also be generated by material loss, detuning and frequency-dependent interactions. This gives a passive reduced-design regime, instead of requiring an idealized gain model.

The perturbation analysis explains the local sensitivity of the resonances near an exceptional point. Under a generic perturbation, an order-$m$ exceptional point splits according to a Puiseux expansion with leading order $\tau^{1/m}$. The numerical experiments are consistent with this prediction across the models considered here: the fitted exponents are close to $1/2$ for dimers and close to $1/3$ for the trimer. The explicit prefactor formula shows that the strength of the splitting depends on the full dispersive derivatives of the reduced pencil, whereas the fractional exponent is determined by the order of the degeneracy.

The higher-dimensional theory shows that the same determinant-based viewpoint extends beyond dimers. In larger resonator arrays, exceptional points are organized by algebraic conditions on $F$ and its frequency derivatives, and their realization becomes an inverse-design problem over the material and geometric parameters. The explicit trimer construction gives a reduced third-order exceptional point with geometric multiplicity one and cube-root splitting.

Several directions remain open. A natural next step is to compare the reduced predictions against the full Lippmann--Schwinger discretization for finite resonator size, in order to quantify the error between the matrix pencil and the original Helmholtz problem. Another direction is to develop systematic inverse-design algorithms for higher-order exceptional points in larger graded arrays. Finally, the sensitivity prefactors derived here suggest a route towards optimizing dispersive resonator systems for sensing, mode conversion and robust exceptional-point tracking under fabrication constraints.

\appendix

\section{Proofs of the Main Statements}
\label{app:proofs}

\subsection{Proof of Theorem~\ref{thm:lippmann-schwinger}}
\label{app:proof-ls}

\begin{proof}
This is the volume-potential formulation of the Helmholtz transmission problem. Applying the outgoing resolvent of $\Delta+\delta^2k_0^2$ to the contrast-supported source term gives the stated integral equation, and applying $\Delta+\delta^2k_0^2$ to the integral representation recovers the Helmholtz equation in the background and in each component, together with the transmission conditions and the outgoing radiation condition. This is precisely the Lippmann--Schwinger representation used for highly dispersive halide perovskite resonators in \cite{alexopoulos2022asymptotic,alexopoulos2023mathematical} and in the graded-array formulation \cite{alexopoulos2026graded}. 
\end{proof}

\subsection{Proof of Lemma~\ref{lem:dominant-profile}}
\label{app:proof-dominant-profile}

\begin{proof}
This is the constant-mode approximation for the dominant self-interaction eigenfunction. In the two-dimensional subwavelength regime, the leading logarithmic part of the volume potential is independent of the spatial variable on each component. Hence the dominant eigenfunction is approximated by the normalized constant function on that component. The estimate
\begin{align*}
  \phi_i
  =
  \widehat 1_{D_i}
  +
  O\bigl(|\log\delta|^{-1}\bigr)
\end{align*}
is the one used in the asymptotic analysis of highly dispersive halide-perovskite resonators and coupled dispersive resonator systems \cite{alexopoulos2022asymptotic,alexopoulos2023mathematical}.
\end{proof}

\subsection{Proof of Lemma~\ref{lem:projected-system}}
\label{app:proof-projected-system}

\begin{proof}
We decompose the homogeneous Lippmann--Schwinger equation into the self-interaction blocks $\mathcal K_i^{\delta k_0}$ and the cross-interaction blocks $\mathcal R_{ji}^{\delta k_0}$. Near the isolated eigenvalue $\nu_i(\delta,\omega)$, the pole-pencil decomposition gives the leading contribution of the inverse self-interaction operator in the direction $\phi_i$. Projecting the equation onto the span of $\phi_1,\ldots,\phi_N$ therefore gives a closed system for the coefficients $q_i$. The diagonal entries contain the factors
\begin{align*}
  1-\delta^2\omega^2\xi_i(\omega)\nu_i(\delta,\omega),
\end{align*}
whereas the off-diagonal entries are the projected interactions
\begin{align*}
  -\delta^2\omega^2\xi_j(\omega)
  \left\langle
    \mathcal R_{ji}^{\delta k_0}\phi_j,
    \phi_i
  \right\rangle.
\end{align*}
This is the finite-dimensional reduction derived in \cite{alexopoulos2022asymptotic,alexopoulos2023mathematical} and the component-dependent, graded version is the formulation used in \cite{alexopoulos2026graded}.
\end{proof}

\subsection{Proof of Theorem~\ref{thm:reduced-resonance-equation}}
\label{app:proof-reduced}

\begin{proof}
This is the finite-dimensional volume-potential reduction derived for coupled highly dispersive resonator systems in \cite{alexopoulos2022asymptotic,alexopoulos2023mathematical}. The graded-array version, including component-dependent geometry and material parameters, is the reduction used in \cite{alexopoulos2026graded}. The argument consists of decomposing the volume-potential equation into self-interaction and cross-interaction blocks, applying the pole-pencil decomposition near the isolated self-interaction mode on each component, and projecting onto the associated resonant vectors. The remaining part of the operator contributes the error term $\mathcal E_\delta$ in the projected equation. Since this perturbation is small uniformly on compact subsets of $W$, the leading resonance condition is $\det\mathcal L(\omega,p)=0$.
\end{proof}

\begin{remark}
    For a simple reduced resonance, the persistence and multiplicity statement follows from Rouch\'e's theorem applied to $\det(\mathcal L+\mathcal E_\delta)$ on the boundary of $B(\omega_0,r)$.
\end{remark}

\subsection{Proof of Lemma~\ref{lem:dispersion-nonlinear-pencil}}
\label{app:proof-dispersion-pencil}

\begin{proof}
As in the proof of Lemma~\ref{lem:projected-system}, the reduced matrix is obtained by projecting the homogeneous Lippmann--Schwinger equation onto the dominant profiles $\phi_i$. At leading order, its entries are of the form
\begin{align*}
  \mathcal L_{ii}(\omega,p)
  =
  1-\delta^2\omega^2\xi_i(\omega,p)
  \left\langle
    \mathcal K_i^{\delta k_0}\phi_i,
    \phi_i
  \right\rangle,
\end{align*}
and, for $i\neq j$,
\begin{align*}
  \mathcal L_{ij}(\omega,p)
  =
  -\delta^2\omega^2\xi_j(\omega,p)
  \left\langle
    \mathcal R_{ji}^{\delta k_0}\phi_j,
    \phi_i
  \right\rangle,
\end{align*}
up to the same normalization conventions used in the reduced graded-array matrix. Thus, the material contrast appears explicitly in the diagonal self-interaction term and in the off-diagonal source interaction terms. The entries are finite sums and products of projected self-interaction coefficients, projected cross-interaction coefficients and material contrasts. Therefore, away from material poles, the holomorphy of these factors implies the holomorphy of $\mathcal L(\cdot,p)$. 
\end{proof}

\subsection{Proof of Lemma~\ref{lem:determinant-criterion}}
\label{app:proof-det-criterion}

\begin{proof}
The order of the zero of $F(\cdot,p_*)$ is, by definition, the algebraic multiplicity of the characteristic value $\omega_*$. The dimension of $\ker\mathcal L(\omega_*,p_*)$ is the geometric multiplicity. Hence the assumptions give
\begin{align*}
  \alg_{\omega_*}\mathcal L(\cdot,p_*)=m,
  \qquad
  \geo_{\omega_*}\mathcal L(\cdot,p_*)=1,
\end{align*}
which is precisely the definition of an exceptional point of order $m$.
\end{proof}

\subsection{Proof of Lemma~\ref{lem:frozen-material}}
\label{app:proof-frozen-material}

\begin{proof}
For any differentiable square matrix $M(t)$, Jacobi's formula gives
\begin{align*}
  \frac{d}{dt}\det M(t)
  =
  \tr\bigl(
    \adj M(t)M'(t)
  \bigr).
\end{align*}
Applying this formula to $M(\omega)=\mathcal L(\omega,p_*)$ and to $M(\omega)=\mathcal L_{\mathrm{fr}}(\omega,p_*)$ gives
\begin{align*}
  \partial_\omega F(\omega_*,p_*)
  =
  \tr\bigl(
    \adj\mathcal L(\omega_*,p_*)
    \partial_\omega\mathcal L(\omega_*,p_*)
  \bigr),
\end{align*}
and the analogous identity for $F_{\mathrm{fr}}$. Since $\mathcal L_{\mathrm{fr}}(\omega_*,p_*)=\mathcal L(\omega_*,p_*)$, their adjugate matrices agree at $(\omega_*,p_*)$. Subtracting the two identities gives the stated formula. In the frozen-material pencil the material contrasts are independent of $\omega$ at the level of the material law, so the terms containing $\partial_\omega\xi_i(\omega_*,p_*)$ are absent.
\end{proof}

\subsection{Proof of Corollary~\ref{cor:exceptional-cluster}}
\label{app:proof-cluster}

\begin{proof}
Since $F(\cdot,p_*)$ has a zero of order $m$ at $\omega_*$ and no other zeros in $\overline{B(\omega_*,r)}$, one has
\begin{align*}
  \min_{\omega\in\partial B(\omega_*,r)}
  |F(\omega,p_*)|>0.
\end{align*}
The estimate in Theorem~\ref{thm:reduced-resonance-equation} gives $\|\mathcal E_\delta(\omega,p_*)\|\to0$ uniformly on the boundary. Hence $F_\delta(\cdot,p_*)$ converges uniformly to $F(\cdot,p_*)$ on $\partial B(\omega_*,r)$. Rouch\'e's theorem therefore implies that $F_\delta$ and $F$ have the same number of zeros in the disc, counted with multiplicity. Since $r$ can be chosen arbitrarily small inside $W$, the zeros of $F_\delta$ converge to $\omega_*$ as $\delta\to0$.
\end{proof}

\subsection{Proof of Theorem~\ref{thm:transversality-inverse-design}}
\label{app:proof-transversality}

\begin{proof}
The exceptional-point equations for a second-order degeneracy are the four real equations
\begin{align*}
  \mathcal T(\omega,a,b)=0.
\end{align*}
By assumption, the Jacobian of $\mathcal T$ with respect to $(\operatorname{Re}\omega,\operatorname{Im}\omega,a)$ is invertible at $(\omega_*,a_*,b_*)$. The implicit function theorem therefore gives unique functions $\omega(b)$ and $a(b)$, defined for $b$ near $b_*$, such that
\begin{align*}
  \mathcal T(\omega(b),a(b),b)=0.
\end{align*}
This is equivalent to $F(\omega(b),a(b),b)=\partial_\omega F(\omega(b),a(b),b)=0$. The conditions $\partial_\omega^2F\neq0$ and $\rank\mathcal L=N-1$ are open under small perturbations along the zero set: the first by continuity, and the second because some $(N-1)\times(N-1)$ minor is nonzero at the original point. Since $F=0$ along the branch, the rank is exactly $N-1$ nearby, and hence the kernel is one-dimensional. Thus, the nearby solutions remain second-order exceptional points.
\end{proof}

\subsection{Proof of Lemma~\ref{lem:dimer-criterion}}
\label{app:proof-dimer-criterion}

\begin{proof}
The first identity is $F(\omega_*,p_*)=0$ and the second identity is $\partial_\omega F(\omega_*,p_*)=0$. The condition $\partial_\omega^2F(\omega_*,p_*)\neq0$ says that this zero has order exactly two. The result therefore follows from Lemma~\ref{lem:determinant-criterion}.
\end{proof}

\subsection{Proof of Lemma~\ref{lem:balanced-dimer}}
\label{app:proof-balanced-dimer}

\begin{proof}
For the balanced dimer, $F=A^2-\Delta^2-\kappa^2$. The first two displayed conditions are $F(\omega_*,p_*)=0$ and $\partial_\omega F(\omega_*,p_*)=0$, up to the common factor $2$ in the first derivative. The third condition says that the zero has order exactly two. Since $\mathcal L(\omega_*,p_*)$ is a nonzero singular $2\times2$ matrix, its kernel is one-dimensional. Lemma~\ref{lem:determinant-criterion} therefore gives the claim.
\end{proof}

\subsection{Proof of Corollary~\ref{cor:pt-threshold}}
\label{app:proof-pt-threshold}

\begin{proof}
With $\Delta_0=i\Gamma$,
\begin{align*}
  F(\omega)
  =
  A(\omega)^2+\Gamma^2-\kappa_0^2.
\end{align*}
Thus, $F(\omega_*)=0$ and $F'(\omega_*)=0$ hold exactly when $A(\omega_*)=0$ and $\Gamma^2=\kappa_0^2$. Moreover,
\begin{align*}
  F''(\omega_*)
  =
  2A'(\omega_*)^2,
\end{align*}
which is nonzero by assumption. At the threshold $\Gamma^2=\kappa_0^2$, the matrix is singular but nonzero, and therefore has a one-dimensional kernel.
\end{proof}

\subsection{Proof of Theorem~\ref{thm:generic-puiseux}}
\label{app:proof-puiseux}

\begin{proof}
The exceptional-point assumptions imply
\begin{align*}
  \Phi(\omega,0)
  =
  a_m(\omega-\omega_*)^m
  +
  O\bigl((\omega-\omega_*)^{m+1}\bigr),
  \qquad
  a_m\neq0.
\end{align*}
The genericity condition gives
\begin{align*}
  \Phi(\omega_*,\tau)
  =
  b\tau+O(\tau^2),
  \qquad
  b\neq0.
\end{align*}
We write
\begin{align*}
  \omega=\omega_*+\tau^{1/m}z.
\end{align*}
Using the Taylor expansion of $\Phi$ near $(\omega_*,0)$, we obtain
\begin{align*}
  \Phi(\omega_*+\tau^{1/m}z,\tau)
  =
  \tau
  \left(
    a_m z^m+b+O(\tau^{1/m})
  \right),
\end{align*}
uniformly for $z$ in bounded sets. The roots of the limiting polynomial $a_m z^m+b$ are simple because $a_m\neq0$ and $b\neq0$. Hence, by Rouch\'e's theorem, or equivalently by the Weierstrass preparation theorem, the zeros of $\Phi(\cdot,\tau)$ near $\omega_*$ are obtained from perturbations of these $m$ roots. This gives
\begin{align*}
  z_j(\tau)
  =
  c_j+O(\tau^{1/m}),
\end{align*}
and therefore
\begin{align*}
  \omega_j(\tau)
  =
  \omega_*
  +
  c_j\tau^{1/m}
  +
  O(\tau^{2/m}).
\end{align*}
\end{proof}

\subsection{Proof of Corollary~\ref{cor:square-root-dimer}}
\label{app:proof-square-root}

\begin{proof}
This is Theorem~\ref{thm:generic-puiseux} with $m=2$ and
\begin{align*}
  a_2
  =
  \frac{1}{2}\partial_\omega^2F(\omega_*,p_*).
\end{align*}
\end{proof}

\subsection{Proof of Lemma~\ref{lem:prefactor-material-derivatives}}
\label{app:proof-prefactor-material}

\begin{proof}
By Corollary~\ref{cor:square-root-dimer}, the square of the leading coefficient is
\begin{align*}
  \mathfrak c^2
  =
  -
  \frac{2b}{\partial_\omega^2F(\omega_*,p_*)},
\end{align*}
where
\begin{align*}
  b
  =
  \left.
  \frac{d}{d\tau}
  F(\omega_*,p_*+\tau h)
  \right|_{\tau=0}.
\end{align*}
Jacobi's formula applied to the parameter derivative gives
\begin{align*}
  b
  =
  \tr\bigl(
    \adj\mathcal L(\omega_*,p_*)
    D_p\mathcal L(\omega_*,p_*)[h]
  \bigr).
\end{align*}
This proves the displayed formula for $\mathfrak c^2$. Applying Jacobi's formula to the frequency derivative gives the identity for $\partial_\omega F$. Differentiating it once more shows that $\partial_\omega^2F(\omega_*,p_*)$ contains the derivatives of $\partial_\omega\mathcal L$, and therefore the terms $\partial_\omega\xi_i$ and $\partial_\omega^2\xi_i$ whenever the material contrasts enter the reduced matrix.
\end{proof}

\subsection{Proof of Lemma~\ref{lem:higher-order-criterion}}
\label{app:proof-higher-order}

\begin{proof}
This is Lemma~\ref{lem:determinant-criterion} applied to the $N$-dimensional reduced matrix. The order of the zero of $F$ is the algebraic multiplicity, while the kernel dimension is the geometric multiplicity.
\end{proof}

\subsection{Proof of Lemma~\ref{lem:trimer-third-order}}
\label{app:proof-trimer}

\begin{proof}
The determinant is computed directly by expanding along the first row:
\begin{align*}
  \det\mathcal L_3(\omega,\tau)
  =
  A(\omega)^3-\kappa^2\tau.
\end{align*}
At $\tau=0$, the assumptions $A(\omega_*)=0$ and $A'(\omega_*)\neq0$ imply that $A(\omega)^3$ has a zero of order three at $\omega_*$. This gives the displayed derivative identities. Finally,
\begin{align*}
  \mathcal L_3(\omega_*,0)
  =
  \begin{pmatrix}
    0 & -\kappa & 0\\
    0 & 0 & -\kappa\\
    0 & 0 & 0
  \end{pmatrix},
\end{align*}
whose kernel is spanned by $(1,0,0)^\top$. Hence the algebraic multiplicity is three and the geometric multiplicity is one. The claim follows from Lemma~\ref{lem:higher-order-criterion}.
\end{proof}

\section{Numerical Reproducibility}
\label{app:numerical-reproducibility}

In this appendix we record the reduced models and numerical procedures used in Section~\ref{sec:numerics}. All experiments are performed on the reduced determinant
\begin{align*}
  F(\omega,p)=\det\mathcal L(\omega,p).
\end{align*}
Second-order exceptional points are computed by solving the four real equations obtained from
\begin{align*}
  F(\omega,p)=0,
  \qquad
  \partial_\omega F(\omega,p)=0.
\end{align*}
The complex frequency is written as $\omega=\omega_r+i\omega_i$, and the remaining real unknowns are chosen according to the experiment. The nonlinear systems are solved by Newton iteration with finite-difference Jacobians. After convergence, we check the residuals, compute the next nonzero frequency derivative of $F$, and evaluate the rank or modal conditioning of $\mathcal L$ near the computed point.

\subsection{Reduced local response and interaction}

The numerical experiments use the reduced determinant associated with the matrix $\mathcal L(\omega,p)$ introduced in the main text. For the halide-perovskite-inspired examples, the material law is the one given in \eqref{eq:hp-permittivity}, inserted through the contrast $\xi_i=\mu_0(\varepsilon_i-\varepsilon_0)$. The local response and the volume-mediated interaction are obtained from the projected self- and cross-interaction coefficients described in Section~\ref{sec:reduced-problem}.

For the model dimers and trimers which are not tied to a specific material law, we use scalar reduced entries chosen to have the same analytic structure as the projected dispersive response. Their exact numerical values are only used to generate the figures in Section~\ref{sec:numerics}; the theoretical claims depend on the determinant conditions and not on these particular constants. In all cases, the frequency-dependent coupling represents the projected Green-function interaction between distinct resonators, with the subwavelength scale and resonator separation fixed during each continuation.

\subsection{Dimer computations}

For the frequency-dependent volume-interaction dimer, the reduced matrix is
\begin{align*}
  \mathcal L(\omega,\Gamma)
  =
  \begin{pmatrix}
    A(\omega)+i\Gamma & -\kappa(\omega)\\
    -\kappa(\omega) & A(\omega)-i\Gamma
  \end{pmatrix},
\end{align*}
where $A$ is the scalar local response and $\kappa$ is the projected interaction coefficient. The frequency $\omega$ and the detuning parameter $\Gamma$ are found from $F=\partial_\omega F=0$. The splitting plot is obtained by replacing $\Gamma_*$ with $\Gamma_*(1+\tau)$ and tracking the two nearby zeros of $F$.

For the passive damped dimer, the unknowns are the complex frequency and two real material parameters of the second resonator, while the damping coefficients are kept positive. The reduced matrix is
\begin{align*}
  \mathcal L(\omega)
  =
  \begin{pmatrix}
    A_1(\omega) & -\kappa(\omega)\\
    -\kappa(\omega) & A_2(\omega)
  \end{pmatrix}.
\end{align*}
The perturbation used to measure the square-root splitting is $\beta_2=\beta_2^*+\tau$, with $\gamma_2=\gamma_2^*$ fixed. The reduced modal condition number is computed from the two null vectors associated with the two nearby resonances.

\subsection{Trimer computation}

The reduced third-order trimer computation uses the normal form
\begin{align*}
  \mathcal L_3(\omega,\tau)
  =
  \begin{pmatrix}
    A(\omega) & -\kappa & 0\\
    0 & A(\omega) & -\kappa\\
    -\tau & 0 & A(\omega)
  \end{pmatrix}.
\end{align*}
The determinant is
\begin{align*}
  F(\omega,\tau)=A(\omega)^3-\kappa^2\tau.
\end{align*}
For each $\tau>0$, the three branches are computed by solving $A(\omega)^3=\kappa^2\tau$. The fitted exponent is obtained by a least-squares fit of the logarithm of the maximum pairwise resonance splitting against $\log\tau$.

\subsection{Inverse design}

For inverse design, the target frequency is
\begin{align*}
  \omega_{\mathrm{tar}}=1.7-0.032i.
\end{align*}
The unknowns are $\beta_2$, $\gamma_2$, the coupling scale $s$ and the separation $R$. These four real parameters are chosen so that $F(\omega_{\mathrm{tar}},p)=0$ and $\partial_\omega F(\omega_{\mathrm{tar}},p)=0$. At the computed point, the smallest singular value of the real Jacobian in Theorem~\ref{thm:transversality-inverse-design} is $4.05\cdot10^{-3}$, confirming that the local inverse-design map is non-degenerate in the numerical example.

\subsection{Halide-perovskite computations}

The halide-perovskite dimer and trimer use the material law \eqref{eq:hp-permittivity}. In the dimer, the second resonator is tuned through two real material parameters while both damping coefficients remain positive. The exceptional point is computed from $F=\partial_\omega F=0$, and the square-root splitting is measured by perturbing one of the material parameters.

In the trimer, three distinct resonators are placed at non-collinear positions and all pairwise projected interactions are retained. The complex frequency and four real material parameters of the second and third resonators are chosen so that
\begin{align*}
  F(\omega,p)=0,
  \qquad
  \partial_\omega F(\omega,p)=0,
  \qquad
  \partial_\omega^2F(\omega,p)=0.
\end{align*}
The cube-root splitting is measured by perturbing one material parameter of the third resonator, with all remaining material parameters fixed.

\subsection{Numerical summary}

Table~\ref{tab:numerical-summary} records the diagnostic quantities used to support the numerical claims in Section~\ref{sec:numerics}. The residual column reports the size of the solved exceptional-point equations, while the last column records the fitted Puiseux exponent obtained from the corresponding splitting plot.

\begin{table}[htbp]
  \centering
  \begin{tabular}{lccc}
    \toprule
    Model & Residual & Nonzero derivative check & Fitted exponent \\
    \midrule
    Volume-interaction dimer
      & $5.2\cdot10^{-15}$
      & double zero
      & $0.5001$ \\
    Passive damped dimer
      & $1.4\cdot10^{-13}$
      & $|\partial_\omega^2F|\approx 40.18$
      & $0.4980$ \\
    Trimer
      & $1.1\cdot10^{-47}$
      & $\partial_\omega^3F\neq0$
      & $0.3338$ \\
    Inverse-designed dimer
      & $1.2\cdot10^{-12}$
      & $|\partial_\omega^2F|\approx 41.33$, $\sigma_{\min}\approx 4.05\cdot10^{-3}$
      & $0.4999$ \\
    Halide-perovskite-inspired dimer
      & $2.6\cdot10^{-13}$
      & $|\partial_\omega^2F|\approx 39.63$
      & $0.4980$ \\
    Halide-perovskite-inspired trimer
      & $2.8\cdot10^{-10}$
      & $|\partial_\omega^3F|\approx 529.54$
      & $0.3386$ \\
    \bottomrule
  \end{tabular}
  \caption{Numerical diagnostics for the reduced exceptional-point computations.}
  \label{tab:numerical-summary}
\end{table}

\bibliographystyle{abbrv}
\bibliography{references}{}

\begin{thebibliography}{10}

\bibitem{alexopoulos2026graded}
K.~Alexopoulos.
\newblock {Frequency sorting in graded dispersive resonator arrays via
  volume-interaction matrices}.
\newblock preprint, hal-05743831, 2026.

\bibitem{alexopoulos2022asymptotic}
K.~Alexopoulos and B.~Davies.
\newblock Asymptotic analysis of subwavelength halide perovskite resonators.
\newblock {\em Partial Differential Equations and Applications}, 3(4):44, 2022.

\bibitem{alexopoulos2023mathematical}
K.~Alexopoulos and B.~Davies.
\newblock A mathematical design strategy for highly dispersive resonator
  systems.
\newblock {\em Mathematical Methods in the Applied Sciences},
  46(14):15883--15908, 2023.

\bibitem{alexopoulos2026non}
K.~Alexopoulos, B.~Davies, and P.~Millien.
\newblock Non-singular hotspots between closely spaced high-index
  nanoparticles.
\newblock {\em arXiv preprint arXiv:2607.24204}, 2026.

\bibitem{ammari2019fully}
H.~Ammari and B.~Davies.
\newblock A fully coupled subwavelength resonance approach to filtering
  auditory signals.
\newblock {\em Proceedings: Mathematical, Physical and Engineering Sciences},
  475(2228):1--17, 2019.

\bibitem{ammari2021high}
H.~Ammari, B.~Davies, E.~O. Hiltunen, H.~Lee, and S.~Yu.
\newblock High-order exceptional points and enhanced sensing in subwavelength
  resonator arrays.
\newblock {\em Studies in Applied Mathematics}, 146(2):440--462, 2021.

\bibitem{ammari2022exceptional}
H.~Ammari, B.~Davies, E.~O. Hiltunen, H.~Lee, and S.~Yu.
\newblock Exceptional points in parity--time-symmetric subwavelength
  metamaterials.
\newblock {\em SIAM Journal on Mathematical Analysis}, 54(6):6223--6253, 2022.

\bibitem{ammari2022wave}
H.~Ammari, B.~Davies, E.~O. Hiltunen, H.~Lee, and S.~Yu.
\newblock Wave interaction with subwavelength resonators.
\newblock In {\em Applied Mathematical Problems in Geophysics: Cetraro, Italy
  2019}, pages 23--83. Springer, 2022.

\bibitem{ammari2024functional}
H.~Ammari, B.~Davies, and E.~Orvehed~Hiltunen.
\newblock Functional analytic methods for discrete approximations of
  subwavelength resonator systems.
\newblock {\em Pure and Applied Analysis}, 6(3):873--939, 2024.

\bibitem{ammari2020close}
H.~Ammari, B.~Davies, and S.~Yu.
\newblock Close-to-touching acoustic subwavelength resonators: eigenfrequency
  separation and gradient blow-up.
\newblock {\em Multiscale Modeling \& Simulation}, 18(3):1299--1317, 2020.

\bibitem{ammari2018mathematical}
H.~Ammari, B.~Fitzpatrick, H.~Kang, M.~Ruiz, S.~Yu, and H.~Zhang.
\newblock {\em Mathematical and computational methods in photonics and
  phononics}, volume 235.
\newblock American Mathematical Soc., 2018.

\bibitem{ammari2009layer}
H.~Ammari, H.~Kang, and H.~Lee.
\newblock {\em Layer potential techniques in spectral analysis}.
\newblock Number 153. American Mathematical Soc., 2009.

\bibitem{bender1998real}
C.~M. Bender and S.~Boettcher.
\newblock Real spectra in non-hermitian hamiltonians having p t symmetry.
\newblock {\em Physical review letters}, 80(24):5243, 1998.

\bibitem{chen2017exceptional}
W.~Chen, {\c{S}}.~Kaya~{\"O}zdemir, G.~Zhao, J.~Wiersig, and L.~Yang.
\newblock Exceptional points enhance sensing in an optical microcavity.
\newblock {\em Nature}, 548(7666):192--196, 2017.

\bibitem{chouhan2020synthesis}
L.~Chouhan, S.~Ghimire, C.~Subrahmanyam, T.~Miyasaka, and V.~Biju.
\newblock Synthesis, optoelectronic properties and applications of halide
  perovskites.
\newblock {\em Chemical Society Reviews}, 49(10):2869--2885, 2020.

\bibitem{el2018non}
R.~El-Ganainy, K.~G. Makris, M.~Khajavikhan, Z.~H. Musslimani, S.~Rotter, and
  D.~N. Christodoulides.
\newblock Non-hermitian physics and pt symmetry.
\newblock {\em Nature Physics}, 14(1):11--19, 2018.

\bibitem{gohberg1971operator}
I.~C. Gohberg and E.~I. Sigal.
\newblock An operator generalization of the logarithmic residue theorem and the
  theorem of rouch{\'e}.
\newblock {\em Mathematics of the USSR-Sbornik}, 13(4):603--625, 1971.

\bibitem{heiss2012physics}
W.~D. Heiss.
\newblock The physics of exceptional points.
\newblock {\em Journal of Physics A: Mathematical and Theoretical},
  45(44):444016, 2012.

\bibitem{hodaei2017enhanced}
H.~Hodaei, A.~U. Hassan, S.~Wittek, H.~Garcia-Gracia, R.~El-Ganainy, D.~N.
  Christodoulides, and M.~Khajavikhan.
\newblock Enhanced sensitivity at higher-order exceptional points.
\newblock {\em Nature}, 548(7666):187--191, 2017.

\bibitem{kato1966perturbation}
T.~Kato.
\newblock {\em Perturbation theory for linear operators}, volume 132.
\newblock Springer, 1966.

\bibitem{lau2018fundamental}
H.-K. Lau and A.~A. Clerk.
\newblock Fundamental limits and non-reciprocal approaches in non-hermitian
  quantum sensing.
\newblock {\em Nature communications}, 9(1):4320, 2018.

\bibitem{li2023exceptional}
A.~Li, H.~Wei, M.~Cotrufo, W.~Chen, S.~Mann, X.~Ni, B.~Xu, J.~Chen, J.~Wang,
  S.~Fan, et~al.
\newblock Exceptional points and non-hermitian photonics at the nanoscale.
\newblock {\em Nature Nanotechnology}, 18(7):706--720, 2023.

\bibitem{miri2019exceptional}
M.-A. Miri and A.~Alu.
\newblock Exceptional points in optics and photonics.
\newblock {\em Science}, 363(6422):eaar7709, 2019.

\bibitem{moiseyev2011non}
N.~Moiseyev.
\newblock {\em Non-Hermitian quantum mechanics}.
\newblock Cambridge University Press, 2011.

\bibitem{park2020symmetry}
J.-H. Park, A.~Ndao, W.~Cai, L.~Hsu, A.~Kodigala, T.~Lepetit, Y.-H. Lo, and
  B.~Kant{\'e}.
\newblock Symmetry-breaking-induced plasmonic exceptional points and nanoscale
  sensing.
\newblock {\em Nature Physics}, 16(4):462--468, 2020.

\bibitem{sutherland2016perovskite}
B.~R. Sutherland and E.~H. Sargent.
\newblock Perovskite photonic sources.
\newblock {\em Nature Photonics}, 10(5):295--302, 2016.

\bibitem{wiersig2014enhancing}
J.~Wiersig.
\newblock Enhancing the sensitivity of frequency and energy splitting detection
  by using exceptional points: application to microcavity sensors for
  single-particle detection.
\newblock {\em Physical review letters}, 112(20):203901, 2014.

\bibitem{wiersig2020prospects}
J.~Wiersig.
\newblock Prospects and fundamental limits in exceptional point-based sensing.
\newblock {\em Nature communications}, 11(1):2454, 2020.

\bibitem{wiersig2020review}
J.~Wiersig.
\newblock Review of exceptional point-based sensors.
\newblock {\em Photonics research}, 8(9):1457--1467, 2020.

\end{thebibliography}

\end{document}